\documentclass[11pt]{article}
\pdfoutput=1
\usepackage[utf8]{inputenc}
\usepackage{fullpage}

\usepackage{microtype}
\usepackage{graphicx}
\usepackage{booktabs} % for professional tables

\usepackage{algorithm}
\usepackage{algorithmic}

\usepackage{hyperref}

\usepackage{amsmath}
\usepackage{amssymb}
\usepackage{mathtools}
\usepackage{amsthm}

\usepackage{physics}
\usepackage{caption}
\usepackage{subcaption}
\usepackage{complexity}

\usepackage[capitalize,noabbrev]{cleveref}

\usepackage{multirow}

\theoremstyle{plain}
\newtheorem{theorem}{Theorem}[section]

\newtheorem{lemma}[theorem]{Lemma}
\newtheorem{corollary}[theorem]{Corollary}
\newtheorem{claim}[theorem]{Claim}%[section]

\theoremstyle{definition}
\newtheorem{definition}[theorem]{Definition}

\theoremstyle{remark}

\newtheorem{example}{Example}
\crefname{problem}{Problem}{Problems}
\crefname{observation}{observation}{observations}

\newcommand{\eps}{\epsilon}

\usepackage[textsize=tiny]{todonotes}
\usepackage{natbib}

\title{Individual Fairness in Graph Decomposition}

\newcommand*\samethanks[1][\value{footnote}]{\footnotemark[#1]}

\author{
  Kamesh Munagala\thanks{Supported by NSF grant CCF-2113798.} \qquad Govind S. Sankar\samethanks[1] \\
  Computer Science Department \\
  Duke University \\
  \texttt{\{munagala,govind.subash.sankar\}@duke.edu} 
}

\makeatletter
\newcommand{\StatexIndent}[1][3]{%
}
\makeatother

\begin{document}
\maketitle




\begin{abstract}
In this paper, we consider classic randomized low diameter decomposition procedures for planar graphs that obtain connected clusters which are cohesive in that close-by pairs of nodes are assigned to the same cluster with high probability. We require the additional aspect of {\em individual fairness} -- pairs of nodes at comparable distances should be separated with comparable probability. We show that classic decomposition procedures do not satisfy this property. We present novel algorithms that achieve various trade-offs between this property and additional desiderata of connectivity of the clusters and optimality in the number of clusters. We show that our individual fairness bounds may be difficult to improve by tying the improvement to resolving a major open question in metric embeddings. We finally show the efficacy of our algorithms on real planar networks modeling congressional redistricting.
\end{abstract}

\section{Introduction}
In this paper, we study a clustering problem of partitioning users located in a metric space into components, each of diameter $O(R)$ for some input parameter $R$. Such partitioning has been widely studied as {\em low diameter decomposition} (LDD) and has myriad applications. In addition to being a stand-alone clustering procedure, they are, for example, used as a subroutine in metric embeddings
\cite{bartal96metricembedding, fakcharoenphol04tree-embedding}.

In this paper, we consider the question of fairness in how different pairs of nodes are separated in the decomposition. In keeping with terminology in the fairness literature, we term this ``individual fairness'' of the decomposition. 

%In this paper, we propose a study of LDD as a clustering procedure. \margingov{Make this sentence better.} As a canonical application, consider

\subsection{Motivation for Individual Fairness}
As motivation, consider the clustering problem of opening schools or libraries and assigning users to these locations. 
%This is closely related to the $k$-center problem where the goal is to open at most $k$ such locations and assign users to these locations, though in this paper, we will not restrict the number of centers opened.
In such assignments, it is desirable  to preserve communities in the sense that nearby users should be assigned to the same cluster. For instance, in Congressional redistricting \cite{altman98redistricting,cohenaddad18redistricting}, it is preferable to keep communities from being fractured \cite{brubach2020pairwise}. Similarly, when assigning students to schools, students who live close to each other prefer to be assigned to the same school. This is termed {\em community-cohesion} \cite{ashlagi2014community-cohesion}. %and is closely related the notion of envy-freeness~\cite{arnsperger94envy}. A user has envy if they are assigned a further away location than their friends who live close by. 
Of course, no deterministic solution can guarantee community cohesion for all user pairs since some close-by user pairs will be deterministically split. In view of this impossibility, we take the approach of \citet{brubach2020pairwise,Chance} and consider solutions that are a distribution over clusterings.  In such a distribution, different pairs of users are separated with different probabilities, and this probability will typically increase with the distance between the pair. This naturally leads to the notion of {\em pairwise envy-freeness} -- pairs of users with comparable distance should be separated with comparable probability. This ensures no pair of users envies a comparable pair in terms of how cohesive the pair is in the decomposition.   

As another motivation, in Congressional redistricting, one pertinent question is to generate a {\em random sample} of representative maps to audit a given map for unfair gerrymandering~\cite{deford2019recombination}. Ideally, in this random sampling distribution, the marginal probabilities of separation of pairs of precincts of comparable distance should be comparable. This means the probability that any pair of precincts is separated should be {\em both} lower and upper bounded by a value that depends on the distance between them. %In Congressional redistricting, there are additional constraints such as bounded number of districts, and we discuss this in \cref{sec:conclusion}.

We term this notion of fairness as {\em individual fairness} due to its connection to a corresponding notion in fair machine learning \cite{DworkIndividual}. The main goal of our paper is to study the existence and computation of LDDs that are individually fair.

\subsection{Randomized LDD in Planar Graphs}
Before proceeding further, we present a formal model for randomized LDD. We consider the canonical case where the users (resp. neighborhoods or other entities) are nodes of a planar graph $G(V,E)$ and the distance function $d(u,v)$ is the unweighted shortest path distance between nodes $u$ and $v$.   We now define an LDD as follows:

\begin{definition} [Low Diameter Decomposition]
\label{def:ldd}
Given a graph $G(V,E)$ with shortest path distance $d(\cdot, \cdot)$, and given a parameter $R$ (where potentially $R \ge n$), a low-diameter decomposition (LDD) is a (possibly randomized) partition of $V$ into disjoint clusters $\pi_1, \pi_2, \ldots, \pi_r$, so that for $j$ and $u,v \in \pi_j$, $d(u,v) = O(R)$.
\end{definition}

In the applications previously mentioned, planar graphs are a natural model. For instance, in political redistricting, nodes correspond to precincts, and edges connect nodes whose corresponding precincts are adjacent. Similarly, in school assignment, nodes correspond to neighborhoods that cannot be split, with edges between nodes corresponding to adjacent neighborhoods. If the LDD partitions the planar graph into connected regions, these regions will also be connected when the nodes are expanded into the corresponding precincts or neighborhoods. This property is an advantage of planar graphs  over point set models, where a contiguous set of nodes need not be connected when they are expanded. Finally, the radius constraint on the LDD ensures the expanded regions are sufficiently compact.

Now consider a pair $(u,v)$ of users separated by distance $d(u,v)$. Several randomized LDD constructions lead to cohesive communities by making the probability that $(u,v)$ are separated an increasing function of $d(u,v)$. This ensures close-by pairs are together with a larger probability than pairs further away. The natural benchmark for this probability is the quantity $\rho_{uv} = \frac{d(u,v)}{R}$, as the following example shows.

\begin{example}
\label{eg:grid}
Consider a grid graph where vertex $v$ with coordinates $(i,j)$ has edges to the vertices $(i-1,j)$, $(i+1,j)$, $(i,j-1)$, and $(i,j+1)$. A simple LDD is to choose $(\ell,m)$ uniformly at random from $[R] \times [R]$ and cut along the $x$-axis at $\{\ldots, \ell - 2R, \ell -R,  \ell, \ell + R, \ldots \}$ and the $y$-axis at $\{\ldots, m - 2R, m -R,  m, m + R, \ldots \}$. Then, the probability any pair of vertices $(u,v)$ is separated is at most $\rho_{uv}$. 
\end{example}

Indeed, for any planar graph, the classic procedure of \citet{KPR} (KPR, see \cref{alg:kpr}) provides a practical LDD method for a planar graph into connected components of diameter $O(R)$ so that the probability of separation of any pair $(u,v)$ is $O(\rho_{uv})$. 

\subsection{Individual Fairness in LDD} 
We now come to the main focus of this paper: individual fairness. For pairs of nodes $(u_1, v_1)$ and $(u_2,v_2)$ with $d(u_1,v_1) \approx d(u_2,v_2)$, both pairs should be separated with comparable probability, so that neither pair envies the other. Ideally, $(u,v)$ are separated with probability $\Theta(\rho_{uv})$. 

In \cref{eg:grid}, the probability that $(u,v)$ are separated is at least $\rho_{uv}/2$, so that this probability is indeed $\Theta(\rho_{uv})$. This may give us hope that the classical KPR procedure also  recreates such an individual fairness guarantee for general planar graphs. Note that the classical analysis implies an upper bound of $O(\rho_{uv})$ on the probability $(u,v)$  is separated, but no corresponding lower bound.  

\paragraph{Existing Methods can Fail.} As our first contribution (\cref{thm:kpr-negative}), we somewhat counter-intuitively show that the KPR procedure (\cref{alg:kpr}) cannot yield an individual fairness guarantee -- for any distance $d$, some pair $(u,v)$ with distance $d(u,v) = d$ is never separated. 

\begin{theorem}[Proved in \cref{sec:kpr-negative}]
\label{thm:kpr-negative}
For any $n\in \mathbb{N}$ and any distance $1\leq d\leq \frac{\sqrt{n}}{3}$, there exists a planar graph $G_d$ and vertices $u,v\in V(G)$ with $d(u,v)=d$ such that \cref{alg:kpr} 
with
%adversarial choices of roots\footnote{We show in \cref{sec:kpr-negative} that a similar result also holds for random choice of roots.} in each component and 
parameter $R > 2d$ never separates $u$ and $v$.
\end{theorem}

This means that though the procedure guarantees that the probability that $(u,v)$ is separated is $O(\rho_{uv})$, the actual probability of separation -- zero -- need not be comparable to this quantity. Therefore, different vertex pairs of comparable distance can be separated with vastly different probabilities, and this holds even when the vertex pair is an edge.

\subsection{Achieving Individual Fairness: Our Results} 
Given this negative result, our main contribution is a set of algorithms for randomized planar graph partitioning that achieve various levels of individual fairness. We study the trade-offs between three natural desiderata in randomized LDDs  -- each part in the decomposition should be a connected component, these components should have non-trivial size, and the partitioning should be individually fair. 

We first state the above desiderata formally:

\begin{description}
\item[(CON):] Each cluster  forms a single connected component. This is important in applications such as school assignment or political redistricting~\cite{duchin}.

\item[(COMP):] The LDD should optimally compress the graph.  For a general unweighted planar graph with $n$ nodes, we say that a procedure satisfies (COMP) if (i) the expected number of components produced is $ O(n/R)$;\footnote{This bound cannot be improved for a line graph.} and (ii) for any outcome of randomness, the algorithm should achieve non-trivial compression: When $R = \frac{D}{\beta}$, where $D$ is the diameter of the graph and $\beta \ge 1$, there always exist $\Omega(\beta)$ components with diameter $\Omega(R)$.
%and $C = O(n/R)$ when the procedure is run on a line graph with $n$ nodes, regardless of the outcome of randomness.

\item[$(f,g)$-(IF):] This condition captures individual fairness. The probability that $(u,v)$ are separated is at most $f(\rho_{uv})$ and at least $g(\rho_{uv})$, for scalar functions $f,g$.
For an arbitrary function $h$, we say that $f(\rho)=O(h(\rho))$ if $f(\rho_{uv})=O(h(\rho_{uv}))$ for every pair $u,v$; similarly for $g$. In the ideal case, as in \cref{eg:grid}, $\Pr[(u,v) \mbox{ is cut}] = \Theta(\rho_{uv})$ for all $u,v$, so that $f(\rho) = O(\rho)$ and $g(\rho) = \Omega(\rho)$ simultaneously.
\end{description}

Note that the KPR procedure (\cref{alg:kpr}) achieves (CON), (COMP), and $f(\rho) = O(\rho)$. However, by \cref{thm:kpr-negative}, it has $g(\rho) = 0$, and is hence not individually fair. 

\paragraph{Fairness for Edges.}  We first show that if we only seek to be individually fair for pairs $(u,v) \in E$, that is, edges or $d(u,v) = 1$, then there is a modification of the KPR procedure via random vertex weights (\cref{alg:kpr-randwts}) that preserves (COMP) and (CON), and guarantees that any edge $(u,v)$ is separated with probability $\Theta\left(\frac{1}{R}\right)$.  We state the theorem in a bit more generality below.
We show a tighter bound for the case of $d(u,v)=1$ since we present experiments on this case in \cref{sec:implement1}.

%\km{Should we write this theorem just for edges, or does this read ok?}
%\gov{I think the theorem statement is fine. The only change i can think of is to modify the text to say "there is a modification of ... and guarantees a small separation probability for every pair of vertices. In particular, this provides individual fairness for edges."
%}

\begin{theorem}[Proved in \cref{sec:alg-randwts}]
\label{thm:alg-randwts}
There is an LDD procedure (\cref{alg:kpr-randwts}) that produces connected components of diameter at most $43 \cdot (R+1)$, satisfies (COMP), and every pair $(u,v) \in V \times V$ satisfies
\begin{align*}
\frac{1}{2R}\leq \Pr[u,v \text{  separated}]\leq \begin{cases}
3\rho_{uv}+\frac{3}{R}  &\text{ if  $d(u,v)>1$}\\
\frac{3}{R}  &\text{ if  $d(u,v)=1$}\\
\end{cases}.
\end{align*}
\end{theorem}

\paragraph{Lower Bound on Individual Fairness.}  Moving on to pairs $(u,v)$ separated by general distances, we show that obtaining an individual fairness guarantee of $\Theta(\rho_{uv})$  resolves a major conjecture in metric embeddings and is hence likely difficult.  Formally, we tie individual fairness in LDD precisely to the distortion of embedding the graph metric into $\ell_1$. Such an embedding is a function $h$ that maps the vertices $v$ of $G$ to $s$-dimensional points $h(v)$ (for some $s$) with the guarantee that the $\ell_1$ distance between $h(u)$ and $h(v)$ is at least $d(u,v)$, and at most $\alpha \cdot d(u,v)$. The parameter $\alpha \ge 1$ is termed the {\em distortion}.

\begin{theorem}[Proved in \cref{sec:equiv}]
\label{thm:equiv}
If there is an LDD procedure for planar graph\footnote{This theorem holds for arbitrary metric spaces. } $G$ that with any parameter\footnote{Note that \cref{def:ldd} allows for the parameter $R$ to be larger than the diameter of $G$. In this regime, the constraint on the diameter is vacuous, and the only requirement is to satisfy the bounds on the separation probabilities.} $R = \omega(n)$ that achieves $f(\rho)/g(\rho) \le \gamma$ and $f(\rho) = O(\rho)$, then there is an $\ell_1$ embedding of $G$ with distortion $O(\gamma)$. 
\end{theorem} 

The best-known distortion bound for embedding a planar graph into $\ell_1$ is $O(\sqrt{\log n})$, and it is a major open question to improve this to a constant. This implies obtaining the ``optimal bound''  -- $f(\rho) = O(\rho)$ and $g(\rho) = \Omega(\rho)$ simultaneously resolves a major conjecture, since by \cref{thm:equiv}, it implies a constant distortion embedding of $G$ into $\ell_1$.

%\margingov{I moved this paragraph from near Theorem 1.7 to here. }

\paragraph{Positive Results for General Distances.} Given the above difficulty, we show algorithms that achieve increasingly stronger individual fairness guarantees as we relax the requirements of (COMP) and (CON). 

First, we show the following theorem that extends \cref{thm:alg-randwts} to derive a lower bound on separation probability for non-adjacent vertex pairs.

\begin{theorem}[Proved in \cref{sec:generallb1}]
\label{thm:generallb1}
    There is an algorithm (\cref{alg:kpr-modified}) that, for any $\epsilon>0$, achieves  (CON), (COMP), $f(\rho) = O(\rho)$, and $g(\rho) = \Omega(\epsilon^2 \cdot \rho^{2+\epsilon})$ for any $d(u,v)$.
\end{theorem}

The above theorem is technically the most intricate.  We extend the classical analysis in~\cite{KPR} by designing a novel modification of their algorithm. Our analysis hinges on showing the existence of a $K_{3,3}$ minor (and hence a contradiction) if a vertex pair $(u,v)$ is not separated with sufficient probability. At a high level, this is similar to the analysis in \cite{KPR}. However, they show such a minor for deriving the diameter of the decomposition. In contrast, in our case, the diameter follows from their argument but we need the minor construction to show the lower bound. This makes the details of our construction different.

We next show in \cref{sec:rand-radius} that the above bound on $g$ can be improved if we  violate (COMP). 

\begin{theorem}[Proved in \cref{sec:rand-radius}]
\label{thm:generallb2}
For any constant $\eps > 0$, there is an LDD procedure (\cref{alg:kpr-rand-radius}) that achieves  (CON), $f(\rho) = O\left(\rho\right)$, and $g(\rho) = \Omega(\eps \cdot \rho^{1+\eps})$. A similar procedure also gives $f(\rho) = O(\rho)$ and $g(\rho) = \Omega\left(\frac{\rho}{\log R}\right)$.
\end{theorem}

The algorithm that shows the above theorem runs the classical KPR algorithm with a random choice of radius parameter. Such a procedure will preserve connectivity, but the random choice of radius leads to sub-optimal compression, even for a line.

We finally consider the implications of violating both (CON) and (COMP).  We show the following theorem. 

\begin{theorem}[Proved in \cref{sec:equiv0}]
\label{thm:equiv0}
For any planar graph and parameter $R$, there is an LDD (\cref{alg:embed})  with $f(\rho) = O(\rho)$ and  $g(\rho) = \Omega\left(\frac{\rho}{\sqrt{ \log R}}\right)$.
\end{theorem}

The algorithm for showing the above theorem  involves embedding the planar graph into the $\ell_2$ metric and projecting the embedding onto a line, reminiscent of locality-sensitive hashing methods for the $\ell_2$ metric. Such an approach does not preserve either connectivity or (COMP).
Note that by \cref{thm:equiv}, improving the above result for $R = \omega(n)$ would make progress on the conjecture of optimally embedding planar graphs into $\ell_1$.

%\gov{This statement "the algorithm .. is independent of KPR" is technically wrong. The Rao paper that the algorithm uses as a subroutine in turn uses KPR as a subroutine.

%The algorithm also works via an $\ell_1$ embedding instead of the $\ell_2$ embedding we have in the current proof. I could find other $\ell_1$ embedding papers (see for eg \cite{filtser20}) with $O(\sqrt{\log n})$ distortion for planar graphs but I didn't immediately find any other $\ell_2$ embedding with that distortion. }

The trade-offs above are summarized in \cref{tab:my-table}. Note that the gap between $f$ and $g$ improves as we relax the other desiderata of (CON) and (COMP).

\begin{table}[htp]
	\centering
 
	\begin{tabular}{@{}c@{}c@{}c@{}c@{}}
		\toprule
		& (CON) & \ (COMP) & $g(\rho)$ \\ 
        \midrule
	KPR (Alg. \ref{alg:kpr})	& \checkmark      &  \checkmark &  $-$ \\ 
        \midrule
	\cref{thm:generallb1} (Alg. \ref{alg:kpr-modified})	& \checkmark   & \checkmark   &   $\Omega(\epsilon^2 \cdot \rho^{2+\epsilon}),\epsilon>0$   \\ 
        \midrule
	\multirow{2}{*}{ \cref{thm:generallb2} (Alg. \ref{alg:kpr-rand-radius}) } 
        & \multirow{2}{*}{\checkmark} & \multirow{2}{*}{ $-$ }  & $\Omega(\epsilon \cdot \rho^{1+\epsilon}), \ \epsilon > 0$ \\ 
      %  \cmidrule(l){4-5} 
	&               &                &  $\Omega\left(\frac{\rho}{\log R}\right)$   \\ 
        \midrule
 	\cref{thm:equiv0} (Alg. \ref{alg:embed} )& $-$    & $-$    & $\Omega\left(\frac{\rho}{\sqrt{\log R}}\right)$ \\
      \bottomrule
	\end{tabular}
	\caption{Trade-offs in Desiderata. In all cases, $f(\rho) = O(\rho)$. The notation $\Omega_\epsilon$ hides terms dependent on $\eps$.
    }
	\label{tab:my-table}
\end{table}

\paragraph{Empirical Results.} We perform experiments on real-world planar graphs modeling precinct connectivity in redistricting applications in the states of North Carolina, Maryland, and Pennsylvania in the United States. We consider individual fairness on edges (\cref{thm:alg-randwts} and \cref{alg:kpr-randwts}). We show an approach to establish that \cref{alg:kpr-randwts} is more individually fair than \cref{alg:kpr} (KPR) across these datasets. Focusing on \cref{alg:kpr-randwts}, we also show that empirically, the radii of the clusters are much smaller than the theoretical bounds (\cref{lem:KPR}). Further, the number of components satisfies the $O(n/R)$ bound on the number of clusters from (COMP) not just in expectation but almost always. % in the real-world graphs we experimented on.
\subsection{Related Work}

\paragraph{Low Diameter Decomposition.} Randomized low diameter decomposition of metric spaces is widely used as a subroutine for embedding the metric space into a simpler space, such as trees or Euclidean metrics, with provably low distortion, meaning that no distance is stretched too much. For general metric spaces, the work of \citet{bartal96metricembedding,fakcharoenphol04tree-embedding} present LDD procedures that achieve a probability of separation $O(\rho_{uv} \cdot \log n)$, where $n$ is the number of points in the metric space. For graphs excluding a $K_r$ minor, which includes planar graphs, the algorithm of \citet{KPR} has been a predominant tool. \citet{ittai19cops}
devise an entirely different framework using `cop decompositions' that yields improved separation probability bounds for large $r$. This, however, does not yield any improvement for planar graphs ($r = 5$), so we focus on the simpler KPR algorithm in this paper.
We refer the reader to the excellent survey by \citet{Matousek2002} for various embedding results and their connection to LDD. 
In contrast, we consider low diameter decomposition as a stand-alone clustering concept and study individual fairness in the separation probabilities it generates and how this fairness trades off with other natural clustering desiderata.

\paragraph{Fairness.} Over the last decade, fairness has become an important algorithmic objective in machine learning \cite{chouldechova2018fairml}. Fair  matchings \cite{huang2016fair}, rankings \cite{celis2017ranking},
and bandit problems \cite{joseph2016fair} have been considered, among others. Following the work of \citet{chierichetti2017fair}, fair clustering has seen much interest \cite{backurs2019scalablecluster,bera19clustering,
chen2019proportionallycluster,esmaeili2020probabilisticcluster}.

Within the literature on fairness, machine learning literature has considered the notion of individual fairness \cite{DworkIndividual}, where the goal is to find a randomized solution (classifier, regression) such that close-by data points see similar outcomes. This has inspired similar models for problems such as auctions \cite{chawla} and assignments \cite{sharan,zeyu}. We extend this literature to consider fairness in separation probability in randomized clustering.

\paragraph{Clustering.} Our problem is closely related to classical $k$-center clustering, where the goal is to open at most $k$ such locations and assign users to these locations. Though this problem is {\sc NP-Hard},  several efficient algorithms are known that achieve a 2-approximation when the points lie in a metric space \cite{hochbaum85kcenter,hsu79kcenterhardness}. Of particular interest is the work of \citet{brubach2020pairwise} (see also \cite{miller13ldd-expo}) that performs randomized $k$-center clustering and guarantees a probability of separation of $O(\rho_{uv} \cdot \log k)$. However, their work does not guarantee individual fairness for pairs of users with comparable distances. We use this line of work as motivation for studying randomized partitioning and seeking individual fairness guarantees. In contrast with $k$-center clustering, we do not enforce a limit $k$ on the number of clusters since no non-trivial fairness guarantee can be obtained with this restriction.\footnote{For example, consider 4 points located on the real line at $0,1,x,x+1$. If we are only allowed to open 2 clusters, then as $x\to \infty$, the pairs of close-by points cannot be separated without drastically affecting the cluster radii.}

\section{The KPR Partitioning Algorithm}
\label{sec:kpr}
For the sake of completeness, we present the KPR algorithm~\cite{KPR}  as \cref{alg:kpr}. This achieves a randomized LDD with the additional properties that: (a) each cluster $\pi_j$ is a connected component; (b) for any $u,v \in \pi_j$, the shortest path  from $u$ to $v$ in the subgraph induced by $\pi_j$ has length $O(R)$; and (c) the probability that $u,v$ belong to different clusters (or are ``separated'') is $O(\rho_{uv})$. Recall that $\rho_{uv} = \frac{d(u,v)}{R}$, where $d(u,v)$ is the shortest path distance (number of edges) between $u$ and $v$ in $G$.  

First, we present some terminology. We term the maximum shortest path length within the subgraph induced by a cluster as its ``diameter''.  For a BFS tree rooted at vertex $r$, we say that $r$ is at level $0$, the children of $r$ are at level $1$, and so on. We say an edge $(u,v)$ is at level $i$ if $u$ is in level $i$ and $v$ is at level $i+1$ or vice versa. We call an iteration of the loop in step~\ref{alg1-step:phase} one `phase'
of \cref{alg:kpr}.

\begin{algorithm}[htbp]
	\caption{KPR Decomposition~\cite{KPR}}\label{alg:kpr}
	\begin{algorithmic}[1]
    \item[\textbf{Input:}] Integer $R$, planar graph $G=(V,E)$.
		\STATE Set $F\gets \{\}$.
		\FOR{$i=\{1,2,3\}$}\label{alg1-step:phase}
		\FOR{every connected component $C$
			of $G\setminus F$}
		\STATE Perform BFS from an arbitrary root $r$ in $C$.
		\STATE
        Sample $k$ uniformly at random from 
        $\{0,1,\ldots R-1\}$.
        
		\FOR{every edge $e$ at level $\ell\equiv k\mod R$ from $r$}\label{alg1-step:choose edges}
		\STATE $F\gets F\cup \{e\}$\label{alg1-step:remove-edges}
		\ENDFOR
		\ENDFOR
		\ENDFOR
	\end{algorithmic}
\end{algorithm}

\Cref{alg:kpr} can also be  viewed in the following way. At step~\ref{alg1-step:choose edges}, define $V_\delta$ as the vertices whose levels lie in $[\delta R+k , (\delta+1)R+k)$ for integers $\delta$. Then, the algorithm can instead remove all edges leaving $V_\delta$ in step~\ref{alg1-step:remove-edges}, for all $\delta$.

%The following theorem captures the key properties of \cref{alg:kpr}. %The bound on the diameter is from~. %We note that the parameter $R$ can possibly be bigger than the diameter of $G$. 
%In the lemma below, we use ``separated'' to mean ``lie in different clusters''.

\begin{theorem}[\cite{KPR,goemans-lecturenotes}]\label{lem:KPR}
	\Cref{alg:kpr} partitions the vertices of $G$ into clusters, each of which is a disjoint connected component of diameter at most $43 \cdot R$. Further, for any $(u,v) \in V \times V$:
	\begin{align*}
		\Pr[(u,v) \text{ are separated}] \leq 3 \cdot \rho_{uv}.
	\end{align*}
\end{theorem}

The procedure satisfies (CON) by definition. \Cref{lem:number} shows that it  satisfies (COMP). In \cref{sec:kpr-negative}, we show that it does not satisfy individual fairness, since there exist graphs where a pair of vertices are never separated.

\begin{lemma} [(COMP)]
\label{lem:number}
For any planar graph with $n$ vertices and diameter $D$, the  expected number of components produced by \cref{alg:kpr} is at most $3n/R$. Further, regardless of the randomness, there always exist $\lfloor\beta/4 \rfloor$ components with diameter at least $R/8$ assuming $R = \frac{D}{\beta}$.
\end{lemma}

\begin{proof}
    Let $C_1, C_2, C_3$ be the random variables denoting the additional number of components generated in phases $1,2,3$ respectively, and let $C = C_1 + C_2 + C_3$.  Consider the first phase of \cref{alg:kpr} and the associated BFS tree. We can pretend the edges not in the tree don’t exist; this only decreases the number of components generated. Let $a_1, a_2, \ldots, a_q$ denote the number of edges in levels $1, 2, \ldots, q$ in the tree. We pad this sequence with zeros so that the $q$ is a multiple of $R$. We have $\sum_{i=1}^q a_i = n-1 \le n$. If a level $j$ is cut, the number of additional components created is precisely $a_j$. Since any level is chosen for cut with probability $1/R$, we have $\E[C_1] \le n/R$. 

Suppose we have components of size $m_1, m_2, \ldots, m_t$ at the end of the first phase. Applying the above argument again to each of these components,  we have $\E[C_2|C_1] \le \sum_{i=1}^t m_i/R \le n/R$, since $ \sum_{i=1}^t m_i \le n$.  A similar argument holds for $\E[C_3]$, completing the proof of the first part by linearity of expectation. 

For the second part, suppose $R = D/\beta$. Then, the BFS tree in the first phase has a depth of at least $D/2$, so that at least $\lfloor \beta/4 \rfloor$ components have a diameter of at least $R/2$. For each of these components, the BFS tree now has a depth at least $R/2$, so at least one component from each has a diameter of at least $R/4$. Similarly for the third phase, where we obtain a diameter at least $R/8$, completing the proof.

\end{proof}

\section{Individual Fairness for Edges}
\label{sec:alg-randwts}
We will now prove \cref{thm:alg-randwts} via \cref{alg:kpr-randwts}. The algorithm is similar to \cref{alg:kpr}, except it adds random weights to each vertex and cuts either above or below the vertex in the BFS tree depending on the vertex weight.

\begin{algorithm}[htbp]
	\caption{KPR algorithm with random vertex weights}\label{alg:kpr-randwts}
	\begin{algorithmic}[1]
 \item[\textbf{Input:}] Integer $R$, planar graph $G=(V,E)$.
		\STATE Set $F\gets \{\}$.
		\FOR{$i=\{1,2,3\}$}
			\FOR{every connected component $C$
			of $G\setminus F$}
			\STATE Perform BFS from an arbitrary root $r$ in $C$.
			\STATE Sample $\theta$ uniformly at
			random from 
            $\{0,1,\ldots R-1\}$.
			\FOR{\label{step:forloop}every vertex $v$ at level $\ell_v\equiv \theta \mod R$}
			\STATE Independently choose $b_v\in \{0,1\}$ 
         uniformly at random.
			\STATE $\ell_v\gets \ell_v+b_v$.
			\ENDFOR
                \STATE   For non-negative integers $q$, define 
                $V_q=\{v \mid   q R +\theta < \ell_v \le (q+1)R+\theta\}$.
			\STATE Let $E_q$ denote edges with one end-point in set
            $V_q$, and let $E' = \cup_{q \ge 0} E_q$.
			\STATE $F\gets F\cup E'$.
			\ENDFOR
		\ENDFOR
	\end{algorithmic}
\end{algorithm}

The diameter guarantee follows from the proof of \cref{lem:KPR}. Each phase of \cref{alg:kpr-randwts} has at most $R+1$ levels in the BFS tree, and we reuse the bounds from \cref{alg:kpr} with parameter $R+1$ instead of $R$. It is easy to check that the procedure satisfies (CON), and a similar argument to \cref{lem:number} shows (COMP).
 
%To get the upper bound on the separation probability, we first note that the vertices of the shortest path between $u,v$ span at most $ d(u,v)+1$ levels of the BFS tree. With probability $1-\frac{d(u,v)+1}{R}$, none of these levels $\ell$ satisfy $\ell \equiv \theta \mod R$. Since there are three phases, a union bound gives us the upper bound for $d(u,v) > 1$.

%We analyze the case where $d(u,v)=1$ case separately. We consider the case where $R>1$ since the $R=1$ case gives us a vacuous bound. 
To upper-bound the separation probability for an edge, consider the BFS tree just before executing Step~(\ref{step:forloop}) in some phase of the algorithm. If  $u,v$ are on the same level, then that level is $\theta \mod R$  with probability $\frac{1}{R}$ in this phase. If this does not happen, $u$ and $v$ will not be separated in this phase.  Suppose $u$ and $v$ are not on the same level. Without loss of generality, assume that $u$ is in level $\ell$ and $v$ is in level $\ell+1$. If $\ell\equiv \theta \mod R$ (which happens with probability $\frac{1}{R}$), then $b_u = 0$ with probability $1/2$ and $(u,v)$ are subsequently separated. This event, therefore, happens with probability $\frac{1}{2R}$. Similarly, if $\ell+1\equiv \theta \mod R$, then $(u,v)$ are subsequently separated if $b_v = 1$, which again happens with probability $\frac{1}{2R}$. In all other events, $(u,v)$ are not separated. Therefore, in each phase, the probability that $u$ and $v$ are separated is at most $\frac{1}{R}$, and the upper bound follows by taking the union bound over three phases. The bound for more general distances is analogous.

We will only consider the first phase to get the lower bound on the probability of separation. Assume without loss of generality that $u$ is at least as close to the root of the BFS tree as $v$. If $u$ and $v$ are at the same level $\ell$ at the beginning of Step~(\ref{step:forloop}), then with probability $\frac{1}{R}$,	$\theta$ is chosen as $\ell \mod R$. In this event,  $u$ and $v$ are separated if $b_u\neq b_v$, which happens with probability $\frac{1}{2}$. This gives us a probability of $\frac{1}{2R}$ that they are separated. Similarly, if $\ell_u<\ell_v$ at the beginning of Step~(\ref{step:forloop}), then with probability $\frac{1}{2R}$, $\theta\equiv \ell_v \mod R$ and $b_u=0$, in which case $u$ and $v$ are  separated. This completes the proof of \cref{thm:alg-randwts}.
\section{Lower Bound on Individual Fairness}%: Proof of \cref{thm:equiv}}
\label{sec:equiv}

We next prove \cref{thm:equiv}. Let  $n = |V|$. Suppose for any parameter $R = \omega(n)$, there is an LDD procedure that produces clusters such that the probability that $u,v$ are separated lies in $\left[ c_1 \cdot \frac{d(u,v)}{R}, c_2 \cdot \frac{d(u,v)}{R} \right]$ where $c_2 = O(1)$. Note that assuming $R = \omega(n)$ means the diameter of the resulting clusters are irrelevant, and all we care about is the separation probability.

Create an embedding $h$ of $G$ into $\ell_1$ as follows. Sample $\eta = m \cdot R$ such partitions $\{\sigma_i\}_{i=1}^{\eta}$, where $m = \omega(n^2)$. For each cluster $C$ in partitioning $\sigma_{i}$, choose $q^i_C = +1/m$ with probability $1/2$ and $q^i_C = -1/m$ with probability $1/2$. For each $v \in C$, set $t_{iv} = q^i_C$. The embedding has $\eta$ dimensions and the $i^{th}$ dimension of $h(v)$ is $t_{iv}$. 

We have the following lemma for any $(u,v) \in V \times V$:

\begin{lemma}
For any $(u,v)$,  with probability $1 - \frac{1}{m}$,
$$ \frac{c_1}{2} \cdot d(u,v) \cdot \le  \norm{h(u)-h(v)}_1 \le 4 c_2 \cdot d(u,v).$$
\end{lemma}
\begin{proof}
Let $h_{ui}$ be the $i^{th}$ coordinate of $h(u)$. In each dimension, if $(u,v)$ are separated, they get different tokens with probability $1/2$, in which case their distance is $2/m$ in that dimension. Otherwise, their distance is zero. This implies
$$ \E[|h_{ui} - h_{vi}|] = \frac{\Pr[(u,v) \mbox{ are separated }]}{m}.$$
Therefore, 
\begin{align*}
     \E [ \norm{h(u)-h(v)}_1 ] &= \eta \cdot \frac{\Pr[(u,v) \mbox{ are separated } ]}{m}\\
     &=c \cdot d(u,v) \ \mbox{ for } c \in [c_1 , c_2].
\end{align*}
Let $\mu_{uv} = \E [ \norm{h(u)-h(v)}_1 ]$. Note now that $ \norm{h(u)-h(v)}_1$ is the sum of $R$ independent random variables each in $\{0,2/m\}$ and further that $\mu_{uv} \ge c_1 \cdot d(u,v) \ge d(u,v) \ge 1$. Using Chernoff bounds, we have:
$$ \Pr[ \norm{h(u)-h(v)}_1 \notin \mu_{uv} \cdot (1 \pm 0.5) ] \le 2 \cdot e^{-m/8} \le \frac{1}{m} $$
for large enough $m$. 
\end{proof}

Since $m = \omega(n^2)$, the previous lemma implies that with constant probability, for all $(u,v) \in V \times V$, we have
$$ \frac{c_1}{2} \cdot d(u,v) \cdot \le  \norm{h(u)-h(v)}_1 \le 4 c_2 \cdot d(u,v). $$
Scaling the coordinates up by $2/c_1$, this yields a non-contractive embedding into $\ell_1$ with distortion $O(c_2/c_1)$, completing the proof of \cref{thm:equiv}.

\section{Individual Fairness for General Distances}
\label{sec:generallb1}
We now prove \cref{thm:generallb1} via  \cref{alg:kpr-modified}. The key idea of the algorithm is to use two cut sequences in each phase of \cref{alg:kpr}.
%Define $\kappa_\epsilon$ to be the random variable such that 
%$$\Pr[\kappa_\epsilon=i]=\frac{(i+1)^{\epsilon}-i^{\epsilon}}{R^{\epsilon}}$$
%for $i\in \{0,1,\ldots,R-1\}$.

\begin{algorithm}[htbp]
	\caption{KPR algorithm with two cuts per phase}\label{alg:kpr-modified}
	\begin{algorithmic}[1]
 \item[\textbf{Input:}] Integer $R$, planar graph $G=(V,E)$.
		\STATE Run the iteration  of \cref{alg:kpr} once (for $i = 1$) on $G$ from an arbitrary root $r_0$. \label{alg:phase0}
		\STATE Let $F$ be the set of removed edges.
		\FOR{$i=\{1,2\}$}\label{alg2-step:phase}
		\FOR{every connected component $C$
			of $G\setminus F$}
		\STATE Perform BFS with root $r_i$ set to be the vertex 
        closest to the previous root $r_{i-1}$ used for %\linebreakdelete
        creating the component, breaking ties 
        arbitrarily.
		\STATE\label{alg2-step:sample-k1k2}Sample $k_1$ uniformly at  random from
		  $\{0,1,\ldots R-1\}$.
        \STATE Sample $\kappa$ from density $\dfrac{\epsilon}{\kappa^{1-\epsilon} R^{\epsilon}}$; set $k_2=\lceil \kappa\rceil$.
		\FOR{every edge $e$ at levels $\ell\equiv k_1\mod R$ or
        $\ell \equiv k_1+k_2\mod R$ from $r_i$}\label{alg2-step:cuts}
		\STATE $F\gets F\cup \{e\}$.
		\ENDFOR
		\ENDFOR
		\ENDFOR
	\end{algorithmic}
\end{algorithm}

To prove \cref{thm:generallb1}, first, note that it satisfies (CON) by definition. The diameter guarantee of $43 \cdot R$ follows directly from \cref{lem:KPR}. Simply observe that if we just consider the cuts induced by $k_1$ in Step~\ref{alg2-step:cuts}, we have effectively run three phases of \cref{alg:kpr}. Because of $k_2$, we make some additional cuts in the second and third phases, which can only decrease the diameter.

\begin{lemma}
$\Pr[u,v \text{ are separated}] = O\left(  \frac{ d(u,v)}{R}\right)$.
\end{lemma}
\begin{proof}
We use union bounds. The three cuts made in the first step and by $k_1$ in Step~\ref{alg2-step:cuts} add $3 \cdot \rho_{uv}$ to the probability of separation. We now bound the probability that $u,v$ are separated by $k_2$ in Step~\ref{alg2-step:cuts}. Consider the path of length $d=d(u,v)$ between $u,v$ in the tree. They are separated only if at least one edge of this path is at some level $\ell\equiv k_1+k_2 \mod R$.
Suppose that the edges of this path lie between levels $\ell_1$ and $\ell_2$ where we know that $\ell_1-\ell_2\leq d-1$. If $\ell_1\equiv k_1+x\mod R$, then the probability that some edge of this path is cut is the probability that $\kappa\in [x,x+\ell_2-\ell_1)\subseteq [x,x+d)$. This happens with probability
\[
\int_{x}^{x+d}\frac{\epsilon}{\kappa R^{\epsilon}} d\kappa=\frac{\epsilon((x+d)^{\epsilon}-x^{\epsilon} )}{R^{\epsilon}}.
\]
Since $k_1$ is chosen uniformly at random, $x$ is also chosen uniformly at random from $\{0,1,\ldots,R-1\}$. Thus, a cut is made with probability
\begin{align*}
    \sum_{x=0}^{R-1} \frac{\epsilon((x+d)^{\epsilon}-x^{\epsilon} )}{R^{1+\epsilon}}&\leq \int_{0}^{R}\frac{\epsilon((x+d)^{\epsilon}-x^{\epsilon} )}{R^{1+\epsilon}} dx\\
    &\leq \frac{\epsilon((R+d)^{1+\epsilon}-R^{1+\epsilon})}{(1+\epsilon)R^{1+\epsilon}}.
\end{align*}
Choosing $\epsilon\leq 1$, this is at most $2 d/R = 2\rho_{uv}$, giving an overall probability of separation of at most $5\rho_{uv}$.
\end{proof}

We now prove that it satisfies (COMP).
\begin{lemma}
\label{lem:number2}
\cref{alg:kpr-modified} satisfies (COMP).
\end{lemma}
\begin{proof}
    The proof is similar to that of \cref{lem:number}. The same arguments follow for the first cut created.
For the cuts in Step~\ref{alg2-step:cuts}, we separately argue below. Let $ a_1, a_2,\ldots, a_q$ denote the number of edges in levels $1,2, \ldots,q$ of the tree. The expected number of additional components generated by the cut corresponding to $k_1$ is at most $\frac{\sum_{i=1}^q a_i}{R}\leq \frac{n}{R}$.

To bound the expected number of components generated by the cut corresponding to $k_2$, observe that the unconditional probability that any edge of the BFS tree is cut is again $1/R$. Formally, the expected number of components is:

\begin{align*}
    &\sum_{x=0}^{R-1}\Pr[k_1=x]\sum_{i=1}^{q} a_i \Pr[k_2\equiv x-i\mod R | k_1 = x]\\
    &=\sum_{i=1}^{q}a_i\sum_{x=0}^{R-1}\frac{1}{R}\Pr[k_2\equiv x-i\mod R | k_1 = x]=\frac{\sum_{i=1}^q a_i}{R} \le \frac{n}{R}
\end{align*}

The second requirement of (COMP) follows similarly as the proof of \cref{lem:number}, with the observation that we only have $\frac{2D}{R}$ cuts per phase in total for any outcome of randomness.
\end{proof}

\newcommand{\const}{59}

We will now show the following theorem, which yields \cref{thm:generallb1} for $d(u,v) \ge \const$.

\begin{theorem}\label{thm:kpr-2cuts}
	For $d(u,v) \ge \const$, \Cref{alg:kpr-modified} satisfies
	\[
	 \Pr[u,v \text{ are separated}] \ge \frac{\epsilon^2}{16} \cdot \left( \frac{\lfloor d(u,v)/\const \rfloor}{R}\right)^{2+2\epsilon}.
	\]
\end{theorem}

For $(u,v)$ with $d(u,v) < \const$,  we first run \cref{alg:kpr-randwts} and note that \cref{thm:alg-randwts} implies 
 $\Pr[u,v \mbox{ separated}] = \Omega(1/R) = \Omega(\rho_{uv}) = \Omega(\rho_{uv}^{2+2\epsilon}). $
We subsequently run \cref{alg:kpr-modified} on each resulting component, hence yielding \cref{thm:generallb1}.

\subsection{Proof of \cref{thm:kpr-2cuts}}

\newcommand{\dhat}{\hat{d}}
\newcommand{\rhohat}{\hat{\rho}}
 
Fix some pair $(u,v)$ of vertices and let $d = d(u,v)$. We assume $d$ is a multiple of $\const$; otherwise, we can use $\const\lfloor\frac{d}{\const}\rfloor$ as a lower bound
for $d$. First, note that we can assume $d \le 43 \cdot R$, else $u,v$ are separated with probability $1$. Define $\rhohat=\frac{\rho_{uv}}{\const}$ and  $\dhat=\frac{d}{\const}$. From above, we get $\rhohat \le 1$. Further, assume w.l.o.g. that $(u,v)$ are not cut in Step~(\ref{alg:phase0}). 

Conditioned on this event, let $G_1$ denote the component to which $(u,v)$ belong.  In the next iteration (Step~(\ref{alg2-step:phase}) for $i=1$), let the root chosen for this component be $r_1$, and the corresponding BFS tree be $T_1$.  Let $d_1(x,y)$ denote the distance between $x,y$ in $G_1$ and let $\ell_1(u) = d_1(r_1,u)$ denote the level of the tree $T_1$ (measured from the root) that $u$ belongs to.

\paragraph{Events $L_1$ and $E_1$:} Let  $L_1$ be the event that $\abs{\ell_1(u)-\ell_1(v)}<\dhat$. If $L_1$ does not happen, then $(u,v)$ are cut in Step~(\ref{alg2-step:phase}) for $i = 1$ with probability at least $\rhohat$ and we are done. 
We proceed assuming event $L_1$ happens and define the first sandwich event $E_1$ as follows. 

\begin{definition}
\label{def:e1}
In event $E_1$, in Step~(\ref{alg2-step:phase}) for $i = 1$, the $(k_1, k_2)$ chosen satisfy:  
\begin{itemize}
\item $k_1\equiv \ell \mod R$ for some $\ell$ satisfying 
$$\min (\ell_1(u),\ell_1(v))\geq \ell \geq \min(\ell_1(u),\ell_1(v))-\dhat$$
\item $k_1+k_2\equiv \ell \mod R$ for some $\ell$ satisfying 
$$\max (\ell_1(u),\ell_1(v))\leq \ell \leq \max(\ell_1(u),\ell_1(v))+\dhat$$
\end{itemize}
\end{definition}

\begin{figure}[htbp]
    \centering
    \includegraphics[width=0.25 \textwidth]{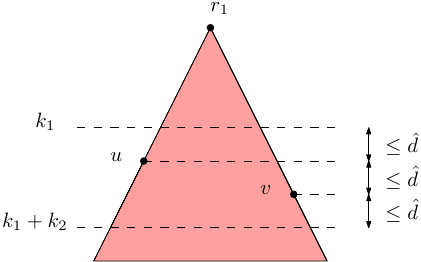}
\caption{An example of the event $E_1\cap L_1$ in $T_1$ (red). Both the cuts made by $k_1,k_2$ are within $2\dhat$ levels of both $u,v$, which are themselves within $\dhat$ levels of each other.}
    \label{fig:firstcuts}
\end{figure}
In other words, the first sandwich event happens when the first cut (corresponding to $k_1$) for that iteration $i = 1$ is made within $\dhat$	levels above the higher of the levels of $u,v$ in $T_1$ and the second cut (corresponding to $k_2$) in that iteration is made within $\dhat$ levels below  the level of the lower of $u,v$ in $T_1$. This is illustrated in \cref{fig:firstcuts}. Observe that for a sandwich event to happen, $k_1$ must take one of $\dhat$ possible values, which happens with probability $\rhohat$,
and $\kappa$ must lie in the range 
$[\max(\ell(u),\ell(v))-k_1,\max(\ell(u),\ell(v))-k_1+\dhat)]$. Since $\max(\ell(u),\ell(v))-k_1 \in [\dhat,2\dhat]$, 
this probability is at least $\displaystyle \frac{(3\dhat)^{\epsilon}-(2\dhat)^{\epsilon}}{R^{\epsilon}} \ge \frac{\epsilon}{4} \cdot \rhohat^{\epsilon}$.
This gives an overall probability of 
$\Pr[E_1] = \frac{\epsilon}{4} \cdot \rhohat^{1+\epsilon}.$ Conditioned on $L_1 \cap E_1$, the pair $(u,v)$ could still be connected. If they are not connected, then we have shown the theorem. Suppose they are still connected, then let the resulting component be $G_2$. Define the root $r_2$ of the BFS tree $T_2$ on $G_2$ analogous to above, and let $\ell_2(x)$ denote the level of node $x \in V(G_2)$ in this tree. 

\paragraph{Events $L_2$ and $E_2$:} We assume the event $L_2$ that $\abs{\ell_2(u)-\ell_2(v)}<\dhat$ happens. Suppose not, then $(u,v)$ are cut in Step~(\ref{alg2-step:phase}) for $i = 2$ with probability at least $\rhohat$. Combining this with the bound for $\Pr[E_1]$ above, we are done. Therefore, assume $L_2$ happens. Then, define the second sandwich event $E_2$ analogous to \cref{def:e1}, with $i=2$ and $\ell_2$ used instead of $i=1$ and $\ell_1$.  We have,   
$$\Pr[E_2 | E_1] = \frac{\epsilon}{4} \cdot \rhohat^{1+\epsilon}\implies  \Pr[E_1 \cap E_2] \ge \frac{\epsilon^2}{16} \cdot \rhohat^{2+2\epsilon}.$$
%where the implication follows from our previous equation.
The following lemma completes the proof of \cref{thm:kpr-2cuts}.

\begin{lemma} [Proved in \cref{app:twocuts}]
\label{lem:kpr-2cuts-main}
Conditioned on $E_1\cap E_2$ (assuming $L_1 \cap L_2$ happens), $u$ and $v$ are always split.
\end{lemma}

We prove \cref{lem:kpr-2cuts-main} by contradiction. Suppose not and that $u,v$ lie in the same connected component after the two phases. Let $G_3$ be this connected component. Our approach will be to  exhibit a $K_{3,3}$ minor in the graph, contradicting the planarity of $G$. This construction is the technical heart of the proof, and we present in \cref{app:twocuts}.

\section{Individual Fairness without (COMP)}
\label{sec:rand-radius}
We now prove \cref{thm:generallb2} via \cref{alg:kpr-rand-radius}. The key idea is to run \cref{alg:kpr} with the parameter drawn from a distribution depending on $R$ instead of the parameter being $R$ directly. 

\begin{algorithm}[htbp]
	\caption{KPR algorithm with random diameter}\label{alg:kpr-rand-radius}
	\begin{algorithmic}[1]
 \item[\textbf{Input:}] Integer $R$, $\alpha\in \mathbb{R}_{\geq 0}$, planar graph $G=(V,E)$
		\STATE Sample $r\in (0,R]$
		according to density $\frac{(1+\alpha)r^{\alpha}}{R^{1+\alpha}}.$
		\STATE Run \Cref{alg:kpr} with the parameter $\lceil r\rceil $ in place of $R$.
	\end{algorithmic}
\end{algorithm}

\begin{lemma} \Cref{alg:kpr-randwts} produces connected components of diameter $O(R)$ and for every $u,v\in V(G)$, satisfies	
\begin{alignat*}{3}	
\frac{\rho_{uv}}{43}& &&\leq \Pr[u,v \text{ are separated}]  \leq 3\rho_{uv}\log(\frac{1}{\rho_{uv}})
\intertext{if $\alpha=0$ and if $\alpha>0$, satisfies}
\left(\frac{\rho_{uv}}{43}\right)^{1+\alpha}& &&\leq \Pr[u,v \text{ are separated}]\leq 3\rho_{uv}\left(1+\frac{1}{\alpha}\right).
\end{alignat*} 
\end{lemma}
\begin{proof}
The diameter bound follows directly from \cref{lem:KPR} since the largest parameter with which we run \cref{alg:kpr} is $R$. Similarly, the connectivity guarantee follows since \cref{alg:kpr} finds connected components.
	
For the lower bound, observe that if $43r<d(u,v)$, then they are separated with probability 1 from the statement of \cref{lem:KPR}. This happens with probability 
\begin{align*}
\int_0^{\frac{d(u,v)}{43}} \frac{(1+\alpha)r^{\alpha}}{R^{1+\alpha}} dr =\left(\frac{d(u,v)}{43R}\right)^{1+\alpha}.
\end{align*}

For the upper bound, we consider them to be separated with probability 1 if $r<3d(u,v)$. If $r\geq 3d(u,v)$, then  they are separated with probability $\frac{3d(u,v)}{r}$ from \cref{lem:KPR}.
This gives us an overall probability of
\begin{align*}
&\int_0^{3d(u,v)} \frac{(1+\alpha)r^{\alpha}}{R^{1+\alpha}} dr+\int_{3d(u,v)}^{R} \frac{3d(u,v)}{r}\frac{(1+\alpha)r^{\alpha}}{R^{1+\alpha}} dr \\
&=\left(\frac{3d(u,v)}{R}\right)^{1+\alpha}
+\frac{3(1+\alpha)d(u,v)}{R^{1+\alpha}}\int_{3d(u,v)}^{R} r^{\alpha-1} dr.
\end{align*}
If $\alpha>0$, the above expression is
\begin{align*}
&\left(3\rho\right)^{1+\alpha}
+\frac{3\rho(1+\alpha)}{\alpha R^{\alpha}}\left(
R^{\alpha} - (3d(u,v))^{\alpha} \right)\leq\frac{3\rho(1+\alpha)}{\alpha},
\end{align*}
 where $\rho = \rho_{uv}$. If $\alpha=0$, the probability is $3\rho +\log(\frac{1}{3\rho}) \leq 3\rho \log(\frac{1}{\rho})$,
completing the proof.
\end{proof}

To obtain the guarantee in \cref{thm:generallb2}, we modify \cref{alg:kpr-rand-radius} as follows:
\begin{enumerate}
\item For constant $\alpha > 0$, we run \cref{alg:kpr} with probability $1-\alpha$ and \cref{alg:kpr-rand-radius} with probability $\alpha$. This yields $g(\rho) = \Omega(\alpha \rho^{1+\alpha})$ and $f(\rho) \le 3 \rho$.
\item When $\alpha  =0$, we run \cref{alg:kpr} with probability $1-\frac{1}{\log R}$ and otherwise run \cref{alg:kpr-rand-radius}. This again yields $g(\rho) = \Omega(\rho/\log R)$ and $f(\rho) = O(\rho)$.
\end{enumerate}

We note that though \cref{alg:kpr-rand-radius} produces $O(n/R)$ components in expectation, it does not satisfy (COMP) since when $r = o(R)$, the maximum diameter is $O(r) = o(R)$.
%it creates $\Omega\left(\frac{n}{R} \cdot \log R \right) = \omega\left(\frac{n}{R}\right)$ components.
\section{Individual Fairness without (CON)}
\label{sec:equiv0}

We now prove \cref{thm:equiv0}. Recall that an embedding is a function $h$ that maps the vertices $v$ of $G$ to $s$-dimensional points $h(v)$ (for some $s$) with the guarantee that the $\ell_1$ distance between $h(u)$ and $h(v)$ is at least $d(u,v)$, and at most $\alpha \cdot d(u,v)$. The parameter $\alpha \ge 1$ is termed the {\em distortion}. The algorithm for embedding planar graphs into $\ell_1$~\cite{Rao-L1} has the following guarantee:

\begin{lemma}[\cite{Rao-L1}]
\label{lem:rao}
Given a planar graph with diameter $\Delta$, there is an embedding $h$ into Euclidean space such that for any pair of vertices $(u,v)$:
$$ d(u,v) \le \norm{h(u) - h(v)}_2 \le O(\sqrt{\log \Delta}) \cdot d(u,v).$$
\end{lemma}

\begin{algorithm}[htbp]
	\caption{LDD via Euclidean Embedding}\label{alg:embed}
	\hspace*{2pt} \textbf{Input}: Integer $R$, planar graph $G=(V,E)$
	\begin{algorithmic}[1]
            \STATE \label{step-kpr1} Run \cref{alg:kpr} to obtain clusters of diameter $O(R)$.
		\FOR{every connected component $C$ found}
                \STATE Embed $C$ into a Euclidean metric via the algorithm 
                %\STATExIndent[1] 
                in~\cite{Rao-L1}. Let $h(v)$ be the embedding of $v \in V$ and let 
                the dimension be $\gamma$.
                \STATE Let $\vec{x}$ be a $\gamma$-dimensional vector whose coordinates 
                %\STATExIndent[1]
                are {\em i.i.d.} $\mathcal{N}(0,1)$.
                \STATE Let $q(v) = h(v) \cdot \vec{x}$.
                \STATE Choose $\theta \in [0,R \cdot \sqrt{\log R}]$ uniformly at random.
                \STATE \label{step7} For $\ell \in \mathbb{Z}$, place all vertices $v$ s.t. 
                %\linebreakdelete
                %\STATExIndent[1]
                $q(v) \in [\theta + \ell \cdot R\sqrt{\log R}, \theta + (\ell+1) R\sqrt{\log R})$ in sub-partition 
                %\STATExIndent[1]
                $C(\ell)$.
                \STATE Output the set $\{C(\ell)\}$ of sub-partitions.
            \ENDFOR
	\end{algorithmic}
\end{algorithm}

The algorithm for showing \cref{thm:equiv0} is shown in \cref{alg:embed}.
Let $d'(u,v) = \norm{h(u) - h(v)}_2$. Note that $\E[\abs{q(v)-q(u)}] = d'(u,v)$, so that 
$$ \Pr[(u,v) \mbox{ are separated in \cref{step7} }] = \frac{d'(u,v)}{R \sqrt{\log R}}.$$
Observe that $(u,v)$ are separated in \cref{step-kpr1} with probability $O(\rho_{uv})$ or  by using \cref{lem:rao} and observing that $\Delta = O(R)$, in \cref{step7} with probability $O(\rho_{uv})$. This shows $f(\rho) = O(\rho)$. The lower bound follows by observing that $(u,v)$ are separated in \cref{step7} with probability at least $\rho/\sqrt{\log R}$ by the above equation.
The algorithm does not satisfy (CON) since any two
vertices have a non-zero probability of being in the same cluster because
of the unbounded support of the normal
distribution.
The algorithm does not satisfy (COMP) either since, for any graph, there is
a non-zero probability of separating
the vertices in $n$ clusters.
The proof of \cref{thm:equiv0} now follows.

\paragraph{General Metrics.} The above algorithm also extends to general metric spaces if we replace \cref{step-kpr1} with the low diameter decomposition procedure in \cite{fakcharoenphol04tree-embedding,miller13ldd-expo}. This guarantees probability of separation $f(\rho) = O(\rho \cdot \log n)$. Now we embed the metric into Euclidean space suffering distortion $O(\log n)$ \cite{Bourgain}, and finally choose $\theta \in [0,R]$ at random before performing \cref{step-kpr1}. This guarantees $g(\rho) = \rho$ and $f(\rho) = O(\rho \cdot \log n)$. 

% !TeX root = ./main.tex
\section{Empirical Results on Precinct Maps}
\label{sec:implement1}

In this section, we run \cref{alg:kpr-randwts} on precinct data obtained from the MGGG Redistricting Lab \cite{mggg-states}. This is naturally modeled as a planar graph  where the vertices correspond to precincts of a state, and edges denote adjacency of precincts. We ran experiments on the states of North Carolina (NC), Pennsylvania (PA), and Maryland\footnote{The MD graph from \citep{mggg-states} is disconnected. It was modified to connect precincts that shared a water border.} (MD). The sizes of the graphs are listed in \cref{tab:graph}.

\begin{table}[htbp]
\centering
\begin{tabular}{@{}cccc@{}}
\toprule
State          & Nodes & Edges & Diameter\\ \midrule
North Carolina & 2692                 &       7593   &68       \\ \midrule
Maryland       &    1809              &         4718     & 66   \\ 
\midrule
Pennsylvania   &   9255               &        25721 & 89         \\ \bottomrule
\end{tabular}
\caption{Sizes of the graphs from \citet{mggg-states}.}
\label{tab:graph}
\end{table}

\begin{figure*}[tb!]
\centering
\begin{subfigure}[t]{0.3\textwidth}
\includegraphics[width=\textwidth]{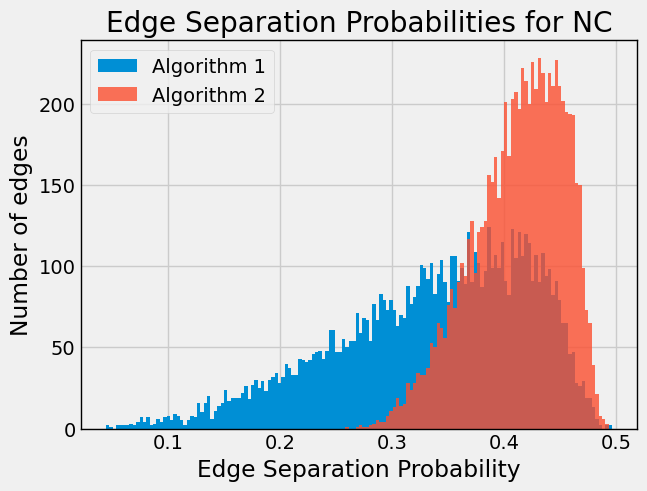}
\caption{\label{fig:gaussian2}}
\end{subfigure}
\begin{subfigure}[t]{0.3\textwidth}
\includegraphics[width=\textwidth]{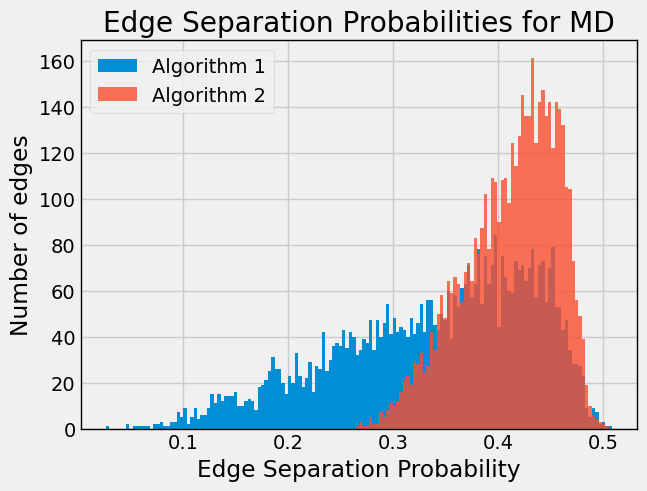}
\caption{\label{fig:gaussianMD} }
\end{subfigure}
\begin{subfigure}[t]{0.3\textwidth}
\includegraphics[width=\textwidth]{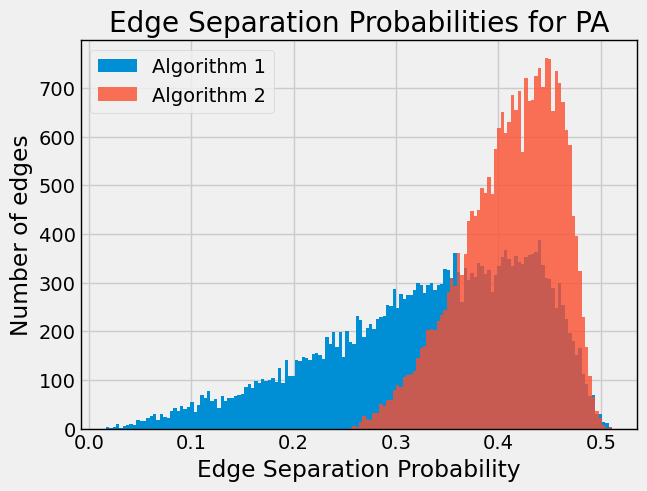}
\caption{\label{fig:gaussianPA}}
\end{subfigure}
\caption{\label{fig:gaussian} Histograms of edge separation probabilities across 3000 runs of \cref{alg:kpr-randwts} on NC, MD, and PA.
}
\end{figure*}

We implemented \cref{alg:kpr-randwts}  with parameter $R=5$
and compared it with \cref{alg:kpr} with the same $R$.
We present the results for individual fairness, as well as the number and diameter of the clusters below.
We omit a detailed runtime analysis of the algorithms as they exclusively entail Breadth-First Search  traversals.

We focus on \cref{alg:kpr-randwts} since the real-world graphs have relatively small distances, and \cref{alg:kpr-modified,alg:kpr-rand-radius} provide weaker guarantees in this regime.
%is that  for the latter algorithms, the lower bound guarantees are extremely weak (for \cref{alg:kpr-modified}) or non-existent (for \cref{alg:kpr-rand-radius}) for small, constant values of $d(u,v)$. An algorithm that combines the random weights from \cref{alg:kpr-randwts} with the extra cuts from \cref{alg:kpr-modified} would obtain a guarantee for small values of distance while also providing a stronger guarantee for larger values of $d(u,v)$.
%(at least $\Omega(R^{\frac{3}{4}}$). However, for graphs with relatively small diameters like the one we experiment on, this modification does not provide any benefits.

\paragraph{Individual Fairness.}
For \cref{alg:kpr-randwts}, for any edge $(u,v) \in E$, \cref{thm:alg-randwts} yields $ \Pr[(u,v) \mbox{ separated}] \in \left(\frac{1}{2R},  \frac{3}{R} \right) = (0.1, 0.6)$, which bounds individual fairness  by a factor of $6$ for this algorithm.  Empirically, we find that \cref{alg:kpr-randwts} outperforms \cref{alg:kpr} significantly on this metric. We plot the histograms of the edge separation probabilities in \cref{fig:gaussian}, showing clearly that \cref{alg:kpr-randwts} is more individually fair than \cref{alg:kpr}.
Additionally, we obtain a stronger guarantee in the lower bound at virtually no cost to the upper bound.
%\Cref{alg:kpr-randwts} has the property that the probabilities of separations of any two edges is at most $\approx\frac{0.5}{0.25}=2$ whereas for KPR, this can be larger than  $\approx\frac{0.5}{0.05}=10$.

In our experiments, we choose the roots for the BFS trees  uniformly at random. Our results show that even with such a strengthened version of  \cref{alg:kpr}, it is more unfair than \cref{alg:kpr-randwts}. This is in line with our result showing that randomness cannot help \cref{alg:kpr} (see \cref{app:additional-impossibility}).

\begin{figure}[htbp]
    \centering
    \includegraphics[width=0.6\linewidth]{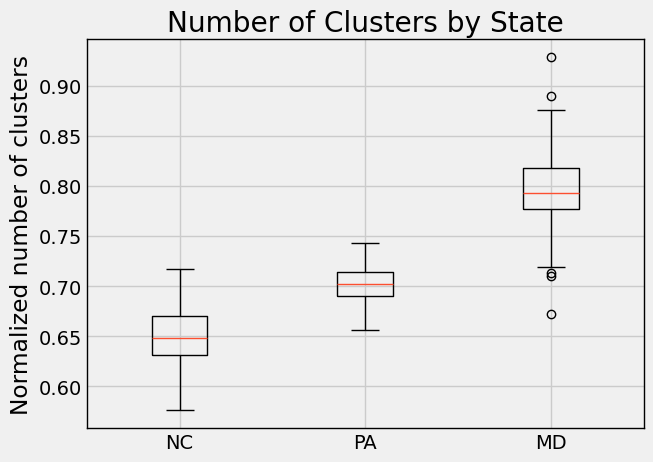}
\caption{Box plot of the number of clusters across 200 runs of \cref{alg:kpr-randwts} on NC, MD, and PA. The \textit{normalized} number of clusters refers to the number of clusters as a fraction of the baseline $\frac{n}{R}$.
}\label{fig:numclusters}
\end{figure}

\paragraph{Number and diameter of clusters.}
For \cref{alg:kpr-randwts}, we can adapt \cref{lem:number} to show that the expected number of clusters is at most $\frac{6n}{R}$. We plot the number of clusters for $200$ runs in  \Cref{fig:numclusters}, and note that the empirically observed value is consistently smaller than $n/R$.

 We next plot the maximum diameter of a cluster in each of the $200$ runs of \cref{alg:kpr-randwts} in \cref{fig:maxdiam}. Although the theoretical guarantee  on the diameter of the largest cluster is $43R = 215$, as stated in \Cref{lem:KPR}, the largest diameter of any cluster is empirically much smaller. 

%For instance, the largest cluster diameter we obtain for PA is $35=7R$. 
%We highlight that \cref{fig:maxdiam} plots the \textit{maximum} diameter in each clustering. Notably, in every clustering, $\gtrsim90\%$ of nodes belong to clusters with diameter less than $4R$.
%, across every state.
%\gov{majority $\to$ exact percentage}
%\km{I think you should show the plot as fraction of $R$, then the text reads ok and we can omit the explanation for PA in the plot. Otherwise it's hard to parse the plot.}

\begin{figure}[htbp]
\centering
    \includegraphics[width=0.6\linewidth]{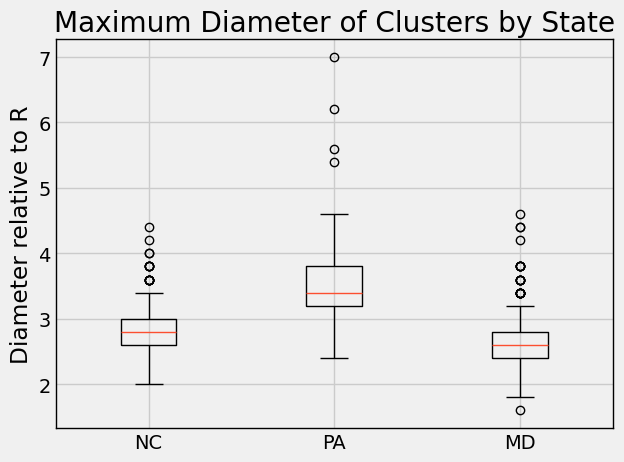}
\caption{Box plot of the maximum diameters of clusters for the states of NC, MD, and PA across 200 clusterings.
The diameters are shown as a fraction
of $R$.
\label{fig:maxdiam}}
\end{figure}

\section{Conclusion}
\label{sec:conclusion}
%In this paper, we introduced the concept of individual fairness for pairs of nodes in randomized graph partitioning, and we presented novel algorithms to achieve various trade-offs between this notion and other desiderata for planar graphs. 
We briefly mention some extensions. Analogous to \citet{KPR}, our results naturally extend to graphs excluding a $K_{r,r}$ minor, with the bounds worsening as $r$ increases.  Similarly, for general metric spaces, we can plug existing low diameter decompositions into \cref{alg:kpr-rand-radius} or \cref{alg:embed}, losing an additional $O(\log n)$ factor in $f(\rho)$ compared to the bounds in \cref{tab:my-table}; (see \cref{sec:equiv0}). % if we plug in the decomposition procedure in \cite{fakcharoenphol04tree-embedding,miller13ldd-expo} into \cref{alg:kpr-rand-radius}, we achieve a decomposition whose diameter is $O(R) $ with high probability, and has $g(\rho) = O\left(\frac{\rho}{\epsilon}\right)$ and $f(\rho) = O(\rho^{1+\epsilon} \log n)$. 
It would be interesting to define analogs for (CON) and (COMP) for general metrics and study their trade-offs with individual fairness, analogous to our results for planar graphs in \cref{sec:alg-randwts,sec:generallb1}. Exploring whether our bounds in \cref{tab:my-table} are improvable for planar graphs would also be an interesting direction. Our work makes edge separation probabilities proportional; defining some notion of Pareto-optimality in conjunction with this would be interesting.

Our techniques can be viewed as randomly compressing a graph while preserving fairness among compressed pairs. This can be used as a pre-processing step to solve more complex optimization problems on the entire graph more efficiently. As an example, in Congressional redistricting, we may want to partition the graph into $k$ parts each with an equal population, while preserving individual fairness. Though no worst-case fairness guarantees are possible for fairness in this setting, one can design efficient heuristics to optimize fairness for the compressed graph, and this is an avenue for future work. Similarly, exploring individual fairness for other optimization problems beyond graph clustering would be interesting.

%When viewed as a clustering result, our work is different from $k$-clustering in that we bound the diameter of the cluster, but not the number of clusters (except in expectation via \cref{lem:number}). Indeed, simple examples show that for bounded number of clusters, no non-trivial worst-case fairness guarantee is achievable. However, one can write the best possible individual fairness achievable for any graph as an integer program. Our techniques can now be viewed as randomly compressing the graph, which helps the efficiency of the integer programs. This is an avenue for future work, since it has implications for generating random redistricting plans, or school assignments.

\paragraph{Acknowledgment.} This work is supported by NSF grant CCF-2113798.

\bibliographystyle{plainnat}
\bibliography{references}

\newpage
\appendix
%\input{content/figures_appendix}
%\section{$g(x) = 0$ for the KPR Algorithm}

\section{Lower Bound for KPR: Proof of \cref{thm:kpr-negative}}
\label{sec:kpr-negative}
%\begin{theorem}\label{thm:kpr-negative}
%	For any $n\in \mathbb{N}$ and $1\leq d\leq \frac{\sqrt{n}}{2}$, there exists a planar graph $G_d$ and vertices $u,v\in V(G)$ with $d(u,v)=d$ such that, if the choices of the roots are adversarial, the Klein-Plotkin-Rao decomposition procedure (\Cref{alg:kpr}) with parameter $R> 2d$ cannot give any non-zero lower bound on the probability of separation of $u,v$.
%\end{theorem}

%\begin{proof}[Proof of \cref{thm:kpr-negative}]
\begin{figure}[htbp]
\centering
\includegraphics[width=0.25\linewidth]{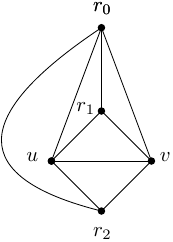}
\caption{\label{fig:counter-d1} Counterexample for $d=1$.}

\end{figure}

\subsection{Proof of \cref{thm:kpr-negative} for $d=1$}
We start with the case of $d=1$ in \cref{fig:counter-d1}.
	We present a constant-sized subgraph, but the arguments remain the same if we attach the graph to an arbitrary star graph with root $r_0$.
    We also consider adversarial choices of roots in \cref{alg:kpr}, but we show later that a similar argument works for random choices of roots.
    Let $r_0$ be the choice of the root. In all subsequent subgraphs, if $r_0$ is not present then $r_1$ is chosen as the root. If neither $r_0$ nor $r_1$ are present, $r_2$ is chosen as the root.
	
	Let $G_0$ be the initial graph and let $T_0$ be the BFS tree from $r_0$. The first observation is that $u,v,r_1,r_2$ belong to the same level of $T_0$.
    Thus, no cut can separate these 
	vertices.
	The only cut that modifies the graph cuts just below $r_0$.
	For the next phase, the component containing $u,v$ is a 4-cycle with a diagonal joining $u,v$.
	Let this graph be $G_1$. 
	
	In $G_1$, $r_1$ is chosen as the new root. The BFS tree $T_1$ from $r_1$
	consists of $r_1$ at level 0,
	$u,v$ at level 1, and $r_2$ at level 2.
	Observe that $u,v$ cannot be separated
	in $G_1$ because they are again at the same level.
	
	Subsequently, since $R>2$, we only have two possible cuts: 
	A cut that separates level 0 from levels 1 and 2, and a cut that separates level 2 from levels 0 and 1.
	In either case, the component
	containing $u,v$ is a 3-cycle with
	$u,v$, and one of $r_1,r_2$.
	Without loss of generality, we consider the latter case so that we get a graph $G_2$
	in which we choose $r_2$ as the
	final root.
	The BFS tree from $r_2$ in $G_2$
	again has $u,v$ at the same level. Hence, $u$ and $v$ are never separated.

\subsection{Proof of \cref{thm:kpr-negative} for $d > 1$}

	\begin{figure}[htbp]
%	\centering
%	\begin{subfigure}[b]{0.49\textwidth}
%		\centering
%		\includegraphics{img/counter-example.pdf}
%		\caption{High-level idea.}
%		\label{fig:counter}		
%	\end{subfigure}
%	\begin{subfigure}[b]{0.49\textwidth}
		\centering
		\includegraphics{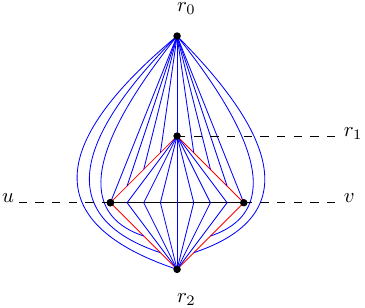}
%		\caption{Construction with more details.}
		
%	\end{subfigure}
	\caption{\label{fig:counter}Counterexample for $d\geq 2$. Each blue line is a path of length $d$.}
\end{figure}

We now prove the statement for general $d\geq 2$. We describe the construction  of the graph we use in \cref{fig:counter}.
	
	\begin{itemize}
		\item Each blue line in the graph is a path
		of length $d$.
		\item The straight line path between $u,v$
		that does not go through $r_1,r_2$ is of length $d$.
		Each vertex in this path is connected to $r_1$, $r_2$ via
		paths of length $d$.
		\item $u,v,r_1,r_2$ are all at distance $d$ from
		$r_0$. 
		\item The red paths between $u,v$ and $r_1,r_2$
		are paths of length $d$. Each vertex 
		in this path is connected to $r_0$
		via a path of length $d$. This is possible
		to do in a planar graph, as shown in the figure.	
	\end{itemize}
	The construction contains $\leq 6d^2+d\leq 9d^2$ vertices. We can arbitrarily attach a star graph with root $r_0$
	to increase the number of vertices.
	However, the number of vertices
	$n$ must satisfy $n\geq 9d^2$.

	Let $r_0$ be the choice of the root.
	In all subsequent subgraphs, if $r_0$ is not present then $r_1$ is chosen as the root. If neither $r_0$ nor $r_1$ are present, $r_2$
	is chosen as the root.
	
	Let $G_0$ be the initial graph and let $T_0$ be the BFS Tree from $r_0$.	
	Observe that by construction, all
	the vertices on the red paths
	at level $d$ in $T_0$. Thus, they must
	all be in the same subgraph as $u,v$ for the second phase.
	Let $G_1$ be the component containing
	$u,v$ in the second phase.
	We argue that $G_1$ is roughly of two `types', depending
	on where the cut is made in the
	first phase:
	\begin{description}
	    \item[Case 1.] The cut is made above level $d$. This ensures that
		all vertices at levels $\geq d$ 
		are in $G_1$.
		Let the resulting graph be $G_1^a$.
		Note that $r_0\not\in V(G_1^a)$.
		\item[Case 2.] The cut is made below
		level $d$.
		This ensures that all vertices
		at levels $\leq d$ are in $G_1$.
		Let the resulting graph be $G_1^b$.
	\end{description}
		
	Note that the $T_0$ has at most $2d$ levels since  all vertices of the graph are within $2d$ levels of $r_0$.
	Furthermore, no cut can remove a vertex both below level $d$ and above level $d$ from $G_1$ because $R>2d$.
	We now consider these two cases,
	ignoring the case where the graph
	is not modified because that can only
	decrease the probability that $u,v$
	are separated.
	The next two lemmas complete the proof of \cref{thm:kpr-negative}.

%	\begin{claim}\label{claim:dangling}
%		Suppose a vertex does not lie on any $(u,v)$, $(r_i,u)$ or $(r_i,v)$ path for $i\in \{0,1,2\}$. Then its removal cannot affect $\Pr[u,v \text{ are separated}]$.
%	\end{claim}
%	\todo{Prove}
	
%	\gov{I need to modify the proof starting from below. I wanted to discuss about what possible lemmas I could add.}

	\begin{figure}[htbp]
		\centering
		\begin{subfigure}[b]{0.49\linewidth}
			\centering
			\includegraphics{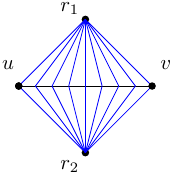}
			\caption{Graph $G_1^a$.}
			\label{fig:phase1a}
		\end{subfigure}
		\begin{subfigure}[b]{0.49\linewidth}
			\centering
			\includegraphics{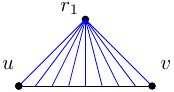}
			\caption{Graph $G_2^a$.}
			\label{fig:phase2a}
		\end{subfigure}
		\caption{Case 1 in the proof of \cref{thm:kpr-negative}.}
	\end{figure}

	\begin{lemma}
		$u$ and $v$ cannot be separated
		by the final two iterations ($i = 2,3$) of \cref{alg:kpr}
		run on $G_1^a$.
	\end{lemma}
	
\begin{proof}
		In this case, the relevant part of the graph is shown in 
	\cref{fig:phase1a}.
	Any vertex not shown is a vertex 
	above level $d$ in $G_0$,
	and does not lie in any path from $r_1$
	to any of the vertices in $G_1^a$.
	We ignore them for ease of exposition,
	but it can be verified that they do not affect our arguments, specifically because they do not affect the positions of $u$ and $v$ in the BFS trees or their connectivity.
	
	Now, $r_1$ is chosen as the root.
	Let $G_2^a$ be the component
	containing $u,v$ in the next phase.
	Like before, only two kinds of cuts can modify the graph.
	\begin{enumerate}
		\item If the cut is made above level $d$,
		$r_2$ lies in $G_2^a$ and $r_1$ is not in $G_2^a$.
		\item If the cut is made below level $d$,
		$r_1$ lies in $G_2^a$ and $r_2$
		is not in $G_2^a$.
	\end{enumerate}
	
	Because of symmetry, we can assume wlog that
	the cut is made below level $d$.
	Regardless of where the cut is, the vertices and edges in \cref{fig:phase2a} are preserved in $G_2^a$.	
	There may be some extra vertices that can safely be ignored again. We again choose
	$r_1$ as the root. Observe that every
	vertex in the straight-line path (colored red) from $u$ to $v$ is at distance $d$
	from $r_1$. Thus, no cut can separate the vertices in this path,
	and $u,v$ must necessarily remain connected after any cut.
\end{proof}

	\begin{figure}[htbp]
		\centering
		\begin{subfigure}[b]{0.49\linewidth}
			\centering
			\includegraphics{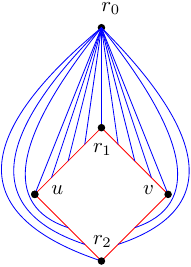}
			\caption{Graph $G_1^b$.}
			\label{fig:phase1b}
		\end{subfigure}
		\begin{subfigure}[b]{0.49\linewidth}
			\centering
			\includegraphics{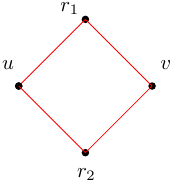}
			\caption{Graph $G_2^b$.}
			\label{fig:phase2b}
		\end{subfigure}
		\caption{Case 2  in the proof of \cref{thm:kpr-negative}.}
	\end{figure}
	
	\begin{lemma}
	$u$ and $v$ cannot be separated
		by the final two iterations ($i = 2,3$) of \cref{alg:kpr}
		run on $G_2^a$.
	\end{lemma}
	
\begin{proof}
		The relevant part of the graph for phase 2 is of the form in
	\cref{fig:phase1b}.
	There are several paths from $r_1$ or $r_2$ to $u$ or $v$ that may or may not be present in $G_1^b$ but are not pictured.
	For our arguments, we only need to observe that the \emph{shortest} paths from $r_1$ and $r_2$ to $u$
	and $v$ are preserved in $G_1^b$.

	In $G_1^b$, making cuts below level $d$
	does not affect any of the relevant $u,v$ paths.
	Thus, we assume that the
	cut is made above level $d$, giving us the graph $G_2^b$.
	The relevant vertices are pictured in
	\cref{fig:phase2b}.
	Like before, we can restrict ourselves
	to this graph because the remaining
	edges do not affect the BFS trees that
	we construct.
	For the next phase,
	we choose $r_1$ as the root.
	The key observation is that
	there is a path from $u$ to $v$
	entirely above level $d$ in this tree through $r_1$, and there is a disjoint
	path from $u$ to $v$ entirely below
	level $d$ through $r_2$.
	We now reuse the
	observation that any cut
	can either remove vertices
	below level $d$ or vertices
	above level $d$, but not both.
	Thus, $u$ and $v$ remain connected
	regardless of the cut.
\end{proof}

\subsection{Additional Impossibility Results}
\label{app:additional-impossibility}

\paragraph{\cref{alg:kpr} with additional iteration.}
The proof of \cref{thm:kpr-negative}
describes a pair $u,v$ that remain
connected after 3 phases of KPR.
One might wonder if it is
possible to do better with additional
phases. Indeed, it turns out that
we can obtain a non-zero probability
of separation. This forms
the basis of \cref{alg:kpr-modified}.
However, we also observe that simple
modifications to \cref{alg:kpr} can only give a very
weak lower bound on the probability
of separation.

Let KPR$_+$ denote \cref{alg:kpr} with one additional phase.
That is, we add one more iteration
to step~\ref{alg1-step:phase}.

\begin{corollary}
For any $n\in \mathbb{N}$ and $1\leq d\leq \frac{\sqrt{n}}{3}$, there exists a planar graph $G_d$ and vertices $u,v\in V(G)$ with $d(u,v)=d$ such that KPR$_+$ with adversarial choices of roots in each component and parameter $R>2d$, separates $u,v$ with probability at most $8\rho_{uv}^4$.
\end{corollary}
\begin{proof}
By running through the proof of \cref{thm:kpr-negative} again,
we can see that if, in any of the phases, the cut is made beyond the first $2d$ levels of the BFS Tree, then $u,v$ will not be separated even in 4 phases. For example, if in the first phase, the cut is made beyond level $2d$ of the BFS Tree, then $u,v$ cannot be separated in the remaining three phases because of \cref{thm:kpr-negative}. Similarly, if this happens in the second or third phases, $u,v$ will never be separated in the remaining phases.
	
Thus, we assume that for the first three phases, the cut is within the first $2d$ levels of the BFS tree. This happens with probability $\left(\frac{2d}{R}\right)^{3}=8\rho_{uv}^3$. The final observation is that $u,v$ are connected by a path of length $d$ in the fourth phase regardless of whether the graph in the third phase is $G_2^a$ or $G_2^b$. Thus, they are at most $d$ levels apart in any BFS tree and can be separated by probability at most $\rho_{uv}$. This gives us an overall probability of separation of at most $8\rho_{uv}^4$.
\end{proof}

\paragraph{Random choices of roots.}
The original algorithm of~\cite{KPR} chooses the nodes arbitrarily. A natural question is whether considering
adversarial	choices of roots is too  pessimistic and whether the `bad' case in \cref{thm:kpr-negative} can be avoided
by a simple scheme such as choosing the root uniformly at random. 

Unfortunately, randomization does not provide any strength beyond possibly providing a $\frac{1}{\poly(n)}$ lower bound. Recall that in the proof of \cref{thm:kpr-negative}, all the adversary needs is to enforce a preference ordering between the roots:  If $r_0$ is present in a component, it must be chosen as the root. If $r_0$ is not present but $r_1$ is, then it must be chosen. Finally, if neither $r_0$ nor $r_1$ is present, $r_2$ should be chosen as the root. We argue that such a preference order can be enforced even in the presence
of uniform sampling of roots.

\begin{corollary}\label{remark:uniform-vs-adversary}
For any $n\in \mathbb{N}$ and $1\leq d\leq n^{\frac{1}{8}}$, there exists a planar graph $G_d$ and vertices $u,v\in V(G)$ with $d(u,v)=d$ such that \cref{alg:kpr} with uniformly random choices of roots in each component and parameter $R > 2d$ separates $u,v$ with probability only $O\left({n^{-\frac{1}{4}}}\right)$.	
\end{corollary}
\begin{proof}
We take the construction in \cref{thm:kpr-negative}. Suppose that it has $m$ vertices. We add star graphs with $m^4$, $m^3$, and $m^2$ nodes with roots $r_0,r_1,r_2$ respectively. Let the total number of nodes now be $n=O(m^4)$. Note that choosing a leaf attached to $r_i$ is equivalent to choosing $r_i$ for the purposes of the proof since the BFS tree from that node is similar to that of the BFS tree from $r_i$.
		
We now follow the proof of \cref{thm:kpr-negative} and argue that the roots are chosen as required with high probability.
In the first phase, the probability that $r_0$ or a leaf attached to it  is not chosen is $\frac{m^3+m^2+m}{m^4+m^3+m^2+m}\leq \frac{1}{m}$. In the second phase, if the graph was not modified in the first phase, the same calculation above applies. Similarly, in $G_1^b$, the probability that $r_0$ is not chosen is again at most $\frac{1}{m}$. In $G_1^a$, the probability that $r_1$ is not chosen is
at most $\frac{m^2+m}{m^3+m^2}\leq \frac{1}{m}$. A similar calculation can be done for both graphs $G_2^a$ and $G_2^b$ in the third phase. By union bounds, the proof is complete.
\end{proof}

%\input{content/number}
%  First note that if $u,v$ are at far apart levels of any of the BFS trees, we argue that they will be split with the required probability. Therefore, we only consider the  case where they are close in level in each of the BFS trees. In this case, we describe an event that occurs with probability $\left(\frac{\rho_{uv}}{100}\right)^4$, and argue that $u,v$ will be separated in this event.

\section{Completing Proof of \cref{thm:kpr-2cuts}: Proof of \Cref{lem:kpr-2cuts-main}}
\label{app:twocuts}

At a high level, our proof outline is similar to that for showing the bounded diameter of clusters in \cref{alg:kpr} in \cite{goemans-lecturenotes},
where the overall goal is to show
a proof by contradiction by exhibiting
a $K_{3,3}$ minor.
We also borrow some of their notation,
though the meaning of the notation is slightly different in our context, as are the details of the construction and analysis. 
For instance, our supernodes that
form the $K_{3,3}$ minor are
easily defined with respect to a few special vertices, 
as opposed to those in \cite{KPR}, which involve many
auxiliary nodes.
We also only need to look at two BFS trees instead of three, as in \cite{KPR}, although we make more cuts.

Another important distinction is that
our analysis inherently requires
the randomness of \cref{alg:kpr-modified}.
For \cref{alg:kpr}, the proof
that the clusters have a small diameter
only hinges on cutting the BFS
tree into slices of width $R$.
In our case, the trees we analyze
possess the necessary properties
only if we condition on the events $E_1,E_2,L_1$ and $L_2$.

\subsection{Special Vertices} 
We first show some simple properties of the graphs $G_1, G_2, G_3$. The first claim follows from the definition of the BFS tree $T_1$, and the occurrence of event $E_1\cap E_2\cap L_1\cap L_2$.

\begin{figure}
    \centering\includegraphics{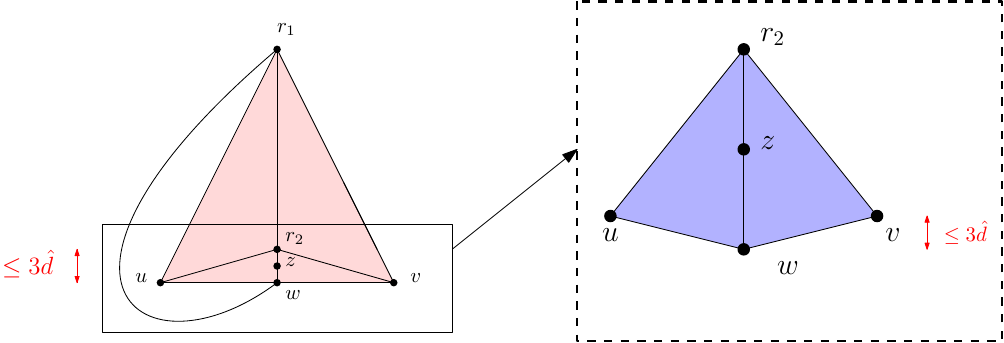}
\caption{ \label{fig:twocuts2}The relative positions of the special vertices in $T_1$ (red) and $T_2$ (blue inset). All of $u,v,w,z,r_2$ are within
$3\dhat$ levels of each other in $T_1$ by virtue of them being in $G_2$. In $T_2$ ,
$u,v,w$ are within $3\dhat$ levels of each other by virtue of them being in $G_3$.
}
   
\end{figure}
 
\begin{claim}\label{claim:levels}
For any $x\in V(G_2),	\abs{\ell_1(x)-\ell_1(u)}$ and $\abs{\ell_1(x)-\ell_1(v)}$ are both at most $2\dhat$.
\end{claim}

\newcommand{\distz}{9\dhat}

Similarly, the next claim follows since $L_2$ happens and $u,v$ are still connected in $G_3$.

\begin{claim}
\label{claim:levels-w}
For any $x\in V(G_3)$, $\abs{\ell_2(x)-\ell_2(u)}$ and $\abs{\ell_2(x)-\ell_2(v)}$ are at most $2\dhat$.
\end{claim}

Since the events $L_1$ and $L_2$ happen, and since $d_i(u,v)\geq d$,
and we defined $\dhat=\frac{d}{\const}$,
we have:

\begin{claim}\label{claim:level-at-least}
For $i\in \{1,2\}$, we have: $\ell_i(u)$ and $\ell_i(v)$ are at least $28\dhat$.
\end{claim}

We now find two special vertices, $w$ and $z$. Let $P$ be the shortest path between $u,v$ in $G_3$. Let $w\in P$ be such that $d_2(u,w)\geq \frac{d}{2}-1>28\dhat$ and $d_2(v,w)\geq \frac{d}{2}-1>28\dhat$.  By \cref{claim:levels-w} and \cref{claim:level-at-least}, we have $\ell_2(w)\geq 25\dhat$. Consider the path from $r_2$ to $w$ in $T_2$.  Let $z$ be a vertex at distance $\distz$ from $w$ along this path. Note that $z$ does not belong to $G_3$. Since $ \ell_2(w)\geq 25\dhat$, we know that $z$ exists.  The construction of these vertices is illustrated in \cref{fig:twocuts2}.

Using \cref{claim:levels-w} with $x = w$ and noting that $\ell_2(z) = \ell_2(w) - \distz$, we obtain: 

\begin{claim}
\label{claim:levels-z}
$\abs{\ell_2(z)-\ell_2(u)}$ and $\abs{\ell_2(z)-\ell_2(v)}$ are at most $11\dhat$.
\end{claim}

\subsection{Construction of $K_{3,3}$ Minor} 
First, the following is a straightforward property of BFS trees that we will use repeatedly.
	
\begin{claim}\label{claim:shortest-path-implies-partition}
For $i\in \{1,2\}$, consider a shortest path $Q$ in $G_i$ from $r_i$ to any vertex $x$. For any $x'\in Q$, if both $x,x'\in G_{i+1}$, then every vertex between $x,x'$ in $Q$ also lies in $G_{i+1}$.
\end{claim}

In the narrative below, we assume $x\in\{u,v,z\}$. Define $P_1(x)$ as the shortest path from $r_1$ to $x$ in $G_1$.  Further, define
\begin{align*}
P_1^-(x)&=\{x'\in P_1(x) \mid x'\in G_1\setminus G_2\}\\
P_1^+(x)&=\{x'\in P_1(x) \mid x'\in G_2\}.
\end{align*}

%This follows from the observation that whether a vertex $x$ lies in $G_2$ is a function of whether $\ell_1(x)$ lies in a fixed interval, and whether it is in the same component as $u,v$. 

Using \cref{claim:shortest-path-implies-partition}, for $x\in \{u,v,z\}$, $P_1^-(x)$ and $P_1^+(x)$ are a partition of $P_1(x)$ into two subpaths. Similarly, for $x\in \{u,v,z\}$, we define $P_2(x)$ as the shortest path from $r_2$ to $x$ in $G_2$. Let
	
\begin{align*}
P_2^-(x)&=\{x'\in P_2(x) \mid \ell_{2}(x')< \ell_2(x)-4\dhat \}\\
P_2^+(x)&=\{x'\in P_2(x) \mid \ell_2(x')\geq \ell_2(x)-4\dhat\}.
\end{align*}

That is, we partition the path into two subpaths based on the distance from $r_2$. The $4\dhat$ vertices closest to $x$ along the path lie in $P_2^+(x)$. Again, by \cref{claim:shortest-path-implies-partition}, these form sub-paths of $P_2(x)$.

Recall that we had assumed for contradiction the existence of a path $P$ from $u$ to $v$ that survived $E_1 \cap E_2$, and that passes through $w$. For $y\in \{u,v\}$, we define $P_w(y)$ as the subpath of $P$ from $w$ to $x$ in $G_2$. Define $P_w(z)$ as the path from $w$ to $z$ in $T_2$. Recall that this is part of the shortest path from $r_2$ to $w$ in $G_2$. 

We again partition the first $4\dhat$ vertices of these three paths from $x \in \{u,v,z\}$ as $P_w^+(x)$ and the remaining part as $P_w^{-}(x)$. Since $P_w(u)$ and $P_w(v)$ are shortest paths in $G_3$, we have
\begin{align*}
P_w^-(x)&=\{x'\in P_w(x) \mid d_{3}(x',w)< d_3(x,w)-4\dhat \}\\
P_w^+(x)&=\{x'\in P_w(x) \mid d_3(x',w)\geq d_3(x,w)-4\dhat\}.
\end{align*}
for $x\in \{u,v\}$. Similarly, since $P_w(z)$ is a shortest path in $G_2$, we have
\begin{align*}
P_w^-(z)&=\{x'\in P_w(z) \mid d_{2}(x',w)< d_2(z,w)-4\dhat \}\\
P_w^+(z)&=\{x'\in P_w(z) \mid d_2(x',w)\geq d_2(z,w)-4\dhat\}.
\end{align*}

We have the following useful properties of these partitions.

\begin{claim}\label{claim:plus-is-not-far}
For $x\in \{u,v,z\}$ and every $y\in \{1,2,w\}$, every vertex in $P_y^+(x)$
is within distance $4\dhat$ of $x$ in $G_2$.
\end{claim}
\begin{proof}
This is true by definition for  $P_2^+(x)$. For $P_3^+(x)$, this also follows from the definition because the distance between pairs can only increase from $G_2$ to $G_3$. Consider some vertex $x'\in P_1^+(x)$. Applying \cref{claim:levels} on both $x'$ and $x$, we get that $\abs{\ell_1(x')-\ell_1(x)}\leq 4\dhat$. But since $P_1^+(x)$ is a shortest path from $r_1$ to $x$,	this means that the path from $x$ to $x'$ can be at most $4\dhat$ long since $\ell_1(x')$ increases linearly along the path from $r_1$  to $x$. Since this entire path survives in $G_2$ (from \cref{claim:shortest-path-implies-partition}), it must also mean that the distance between $x'$ and $x$ in $G_2$ must also be at most $4\dhat$.
\end{proof}

We finally describe the $K_{3,3}$ minor. We set

\begin{itemize}
\item $A_1=P_1^-(u)\cup P_1^-(v)\cup P_1^-(z)$
\item $A_2=P_2^-(u)\cup P_2^-(v)\cup P_2^-(z)$
\item $A_3=P_w^-(u)\cup P_w^-(v)\cup P_w^-(z)$
\item $A_4=P_1^+(u)\cup P_2^+(u)\cup P_w^+(u)$
\item $A_5=P_1^+(v)\cup P_2^+(v)\cup P_w^+(v)$
\item $A_6=P_1^+(z)\cup P_2^+(z)\cup P_w^+(z)$
\end{itemize}
	
The following lemma will complete the proof of \cref{lem:kpr-2cuts-main}.

\begin{lemma}
 \label{lem:K33}
 After contracting each $A_i$, $(\{A_1,A_2,A_3\},\{A_4,A_5,A_6\})$ form a $K_{3,3}$ graph.
 \end{lemma}

\subsection{Proof of \cref{lem:K33}}
The rest of the section is devoted to proving this lemma. First, it can be seen that each $A_i$ is a connected	subgraph. 	
It is also clear that the required edges for the contracted graph to be a $K_{3,3}$ are present.
We can think of representing
the supernodes $A_1$ through $A_6$
by the vertices $r_1,r_2,w,u,v$, and $z$
respectively. The path between
$r_1$ and $u$ is partitioned
between $A_1$ and $A_4$, and hence, the
supernodes have an edge between them.
Similarly, the path between $r_2$
and $u$ is partitioned between $A_2$
and $A_4$, and the path between $w$
and $u$ is partitioned between $A_3$
and $A_4$.

It remains to be argued that each $A_i$ is disjoint from the others. The next set of claims will show that, hence completing the proof of \cref{lem:K33}. 
	
\begin{claim}\label{lem:a1-disj-aj}
$A_1$ is disjoint from every other $A_j$.
\end{claim}	
\begin{proof}
This follows because $A_1$ is the only set that lies in $G_1\setminus G_2$.
The others all lie in $G_2$.	
\end{proof}

\begin{claim}\label{lem:a4-a5-a6-disj}
$A_4,A_5,A_6$ are all mutually disjoint.
\end{claim}
\begin{proof}
Assume for a contradiction that there is some $x\in A_4\cap A_5$. Then $d_2(u,v)\leq d_2(u,x)+d_2(x,v)\leq 8\dhat$ from \cref{claim:plus-is-not-far} which is a contradiction since $d_2(u,v)\geq d$.
Instead, if there is some $x\in A_4\cap A_6$, then we have $d_2(u,w)\leq d_2(u,x)+d_2(x,z)+d_2(z,w)\leq 13\dhat$ which is a contradiction since $d_2(u,w)\geq \frac{d}{2}$. Similarly if there is some $x\in A_5\cap A_6$.	
\end{proof}
 
\begin{claim}\label{lem:a2-disj-a4a5a6}
$A_2$ is disjoint from $A_4,A_5, A_6$.
\end{claim}
\begin{proof}
We first prove this for $A_4$. The proof is symmetric for $A_5$.
Assume for a contradiction
that there is some $x\in A_2\cap A_4$.
We have 3 cases depending
on what part of $A_2$
that
$x$ lies in.
\begin{itemize}
\item $x\in P_2^-(u)$.
$P_2^-(u)$ is disjoint from $P_2^+(u) $ by definition
so $x$ cannot belong to it.  By definition of $P_2^-(u)$ and because $P_2(u)$ is a shortest path from $r_2$ to $u$, we have	$d_2(x,u)> 4\dhat$. However, this means that it cannot belong to $P_1^+(u)$ or $P_w^+(u)$ because it would contradict \cref{claim:plus-is-not-far}.
			
\item
	$x\in P_2^-(v)$.
 Since $x\in A_4$, \Cref{claim:plus-is-not-far} implies
			that $d_2(x,u)\leq 4\dhat$
			which implies $\abs{\ell_2(x)-\ell_2(u)}\leq 4\dhat$.
	But
since event $L_2$
happened, this then implies that
	$\abs{\ell_2(x)-\ell_2(v)}\leq 5\dhat$.
    Since $x$ lies on the shortest path
    from $r_2$ to $v$, this implies that
    $d_2(x,v)\leq 5\dhat$.
    Thus, we have $d_2(u,v)\leq d_2(x,v)+d_2(x,u)\leq 9\dhat$ which is a contradiction.
    
    \item $x\in P_2^-(z)$.
This is similar to the case of $P_2^-(v)$. 
     Since $x\in A_4$, \Cref{claim:plus-is-not-far} implies
    that $d_2(x,u)\leq 4\dhat$
    which also implies $\abs{\ell_2(x)-\ell_2(u)}\leq 4\dhat$.
    But \cref{claim:levels-z}
    then implies that
    $\abs{\ell_2(x)-\ell_2(z)}\leq 15\dhat$.
    Since $x$ lies on the shortest path
    from $r_2$ to $z$, this implies that
    $d_2(x,z)\leq 15\dhat$.
    Thus, we have $d_2(u,w)\leq d_2(x,z)+d_2(x,u)\leq 19\dhat$ which is a contradiction since
    we defined $w$ with
    $d_2(u,w)\geq27\dhat$.
\end{itemize}

		We now prove it for $A_6$. The proof is similar.
  Assume for a contradiction
that $x\in A_2\cap A_6$.
We again have 3 cases depending
on what part of $A_2$
that
$x$ lies in.
		\begin{itemize}
			\item $x
   \in P_2^-(z)$ is disjoint from $P_2^+(z) $
			by definition.
			By definition
			of $P_2^-(z)$ and because $P_2(z)$
			is a shortest path from $r_2$ to $z$,
			we have	$d_2(x,z)> 4\dhat$.
			However, this means that it cannot belong
			to $P_1^+(z)$ or $P_w^+(z)$
			because it would
			contradict \cref{claim:plus-is-not-far}.
			
			\item $x\in P_2^-(v)$.
			Since it
			also belongs to $A_6$, \Cref{claim:plus-is-not-far} implies
			that $d_2(x,z)\leq 4\dhat$
			which also implies $\abs{\ell_2(x)-\ell_2(z)}\leq 4\dhat$.
			But \cref{claim:levels-z}
			then implies that
			$\abs{\ell_2(x)-\ell_2(v)}\leq 15\dhat$.
			Since $x$ lies on the shortest path
			from $r_2$ to $v$, this implies that
			$d_2(x,v)\leq 15\dhat$.
			Thus, we have $d_2(w,v)\leq d_2(x,v)+d_2(x,z)+d_2(z,w)\leq 28\dhat$ which is a contradiction
            to $d_2(u,w)> 28\dhat$.
			
			\item $x\in P_2^-(u)$. This is symmetric to the case of
			 $x\in P_2^-(v)$.
		\end{itemize}	
  This completes the proof.
	\end{proof}

	\begin{claim}\label{lem:a2-disj-a3}
		$A_2$ is disjoint from
		$A_3$.
	\end{claim}
	\begin{proof}
		This primarily follows by observing
		the levels of the vertices in $G_2$.
        Assume for a contradiction
        that there is some $x\in A_2\cap A_3$.
        We have 3 cases depending on what
        part of $A_2$ that $x$
        lies in.

\begin{itemize} 
		\item $x\in P_2^-(z)$.
        Note that $x$ cannot also lie in $P_w^-(z)$
		because $P_2(z)$ and $P_w(z)$ are partitions of
		$P_2(w)$.
        We also have
		$\ell_2(x)<\ell_2(z)-4\dhat=\ell_2(w)-13\dhat$ by definition. By applying \cref{claim:levels-w} on $w$ and combining this
		with the previous inequality, we get
		$\ell_2(x)<\ell_2(u)-11\dhat$.
        This implies that
        $x$ is not in
        $G_3$.
		If $x$ also lied in $P_w^-(u)$
		or $P_w^-(v)$, this would contradict
		\cref{claim:levels-w} because $P_w(u)$
		and $P_w(v)$ are paths in $G_3$.

      \item $x\in P_2^-(u)$. By definition,
		we have
		$\ell_2(x)<\ell_2(u)-4\dhat$.
		If $x$ also lied in $P_w^-(u)$
		or $P_w^-(v)$, this would contradict
		\cref{claim:levels-w} because $P_w(u)$
		and $P_w(v)$ are paths in $G_3$.
  
        If $x\in P_w^-(z)$,
		then $d_2(x,w)\leq \distz$ because
		$P_w(z)$ is only that long. Since $P_w(z)$
		is part of the shortest path from $r_2$ to $w$,
		we also have $\abs{\ell_2(x)-\ell_2(w)}\leq \distz$.
		Applying \cref{claim:levels-w} on $w$,
		we get
		$\abs{\ell_2(x)-\ell_2(u)}\leq 11\dhat$.
		But since $x$ also lies on the shortest
		path from $r_2$ to $u$, we also get
		$d_2(x,u)\leq 11\dhat$.
		Thus, we get $d_2(w,u)\leq d_2(w,x)+d_2(x,u)\leq 20\dhat$ which is a contradiction.
        \item $x\in P_2^-(v)$
        follows similarly by symmetry.
\end{itemize}
	\end{proof}
 
	\begin{claim}\label{lem:a3-disj-a4a5a6}
		$A_3$ is disjoint from $A_4,A_5,A_6$.
	\end{claim}
 	\begin{proof}
		We first show that
		$A_3$ is disjoint from $A_4$.
		The proof that it is disjoint from $A_5$
		follows by symmetry.
		Assume for a contradiction that there is some
$x\in A_3\cap A_4$.
  We have 3 cases depending on
  what part of $A_3$ that $x$ lies in.
		\begin{itemize}
			\item $x\in P_w^{-}(u)$. By definition, it is disjoint from $P_w^{+}(u)$.
			Suppose that $x$
			also lied in $P_1^+(u)$ or $P_2^+(u)$.
			This means that $x\in V(G_3)$.
			By \cref{claim:shortest-path-implies-partition},
			this means that every vertex between
			$x$ and $u$ in $P_1^+(u)$ (or $P_2^+(u)$) also
			lies in $G_3$ and hence, $d_3(u,x)\leq 4\dhat$ by \cref{claim:plus-is-not-far}.
			This is a contradiction because
			$x\in P_w^-(u)$ implies that
			$d_3(x,u)>4\dhat$
			by definition.
			
			\item $x\in P_w^-(v)$. It must be disjoint
			from $P_w^+(u)$ because they are both disjoint
			subsets of the path $P_w(u)\cup P_w(v)$
			from $u$ to $v$.
			Suppose $x$ also lied in
			$P_1^+(u)$ or $P_2^+(u)$.
			By \cref{claim:shortest-path-implies-partition},
			this means that every vertex between
			$x$ and $u$ in $P_1^+(u)$ (or $P_2^+(u)$) also
			lies in $G_3$ and hence, $d_3(u,x)\leq 4\dhat$ by \cref{claim:plus-is-not-far}.
			This contradicts the fact that $P_w(u)\cup P_w(v)$ was a shortest path in $G_3$
			because the path $u\to x\to v$
			is strictly shorter, because $d_3(u,x)\leq 4\dhat<\frac{d}{2}\leq \abs{P_w(u)}$
			and $x\to v$ is a subpath of $P_w(v)$.
			\item $x\in P_w^-(z)$.
			Suppose $x$ also lied in $A_4$.
			By \cref{claim:plus-is-not-far},
			$d_2(u,x)\leq 4\dhat$.
			Thus, we get $d_2(u,w)\leq d_2(u,x)+d_w(x,w)\leq 13\dhat$
			since $x$ lies on $P_w(z)$ which has length
			$\distz$. This is a contradiction.
		\end{itemize}
		
    Now, we show that $A_3$ is disjoint from 
		$A_6$.
  Again, suppose there is some $x\in A_3\cap A_6$. We have 3 cases.
		\begin{itemize}
			\item $x\in P_w^{-}(z)$. By definition, it is disjoint from $P_w^{+}(z)$.
			We have $d_2(x,z)>4\dhat$
			by definition.
			If $x$
			also lied in $P_1^+(z)$ or $P_2^+(z)$,
			then this contradicts \cref{claim:plus-is-not-far}.
			
			\item $x\in P_w^-(u)$.
			If $x\in P_w^-(u)$,
			then $x\in V(G_3)$ and hence
			\cref{claim:levels-w} implies that
			$\abs{\ell_2(x)-\ell_2(u)}\leq2\dhat$.
			But if $x\in A_6$, then \cref{claim:plus-is-not-far}
			implies that $d_2(x,z)\leq 4\dhat\implies
			\abs{\ell_2(z)-\ell_2(x)}\leq 4\dhat$.
			We also know that
			$\ell_2(z)=\ell_2(w)-\distz$.
			Thus, it must be 
			$\abs{\ell_2(w)-\ell_2(x)}\geq 5\dhat$.
			Applying \cref{claim:levels-w}
			on $w$ gives us $\abs{\ell_2(x)-\ell_2(w)}\geq 3\dhat$.
			This contradicts the earlier inequality.
   \item  
			 $x\in P_w^-(v)$ follows by symmetry.
		\end{itemize}
  This completes the proof.
	\end{proof}

These claims complete the proof of \cref{lem:K33}, hence that of \cref{lem:kpr-2cuts-main}, and hence \cref{thm:kpr-2cuts}.

%This follows because we know that the vertices in $G_2$ are precisely those that were not removed in the first sandwich event. The observation then follows from the definition of a sandwich event and \cref{claim:gap-small-g1}.

	%\begin{claim}\label{claim:minus-is-far}
	%	For $x\in \{u,v,z\}$ and every $y\in \{1,2,w\}$, every vertex in $P_y^-(x)$
	%	is not within distance $4\dhat$ of $x$
	%	in $G_2$.
	%\end{claim}

\end{document}